\documentclass{lmcs}
\pdfoutput=1
\usepackage[utf8]{inputenc}

\usepackage{lastpage}
\lmcsdoi{22}{4}{26}
\lmcsheading{}{\pageref{LastPage}}{}{}%
{Mar.~25,~2026}{Dec.~03,~2026}{}


\usepackage{xspace}
\usepackage{booktabs}
\newcommand{\G}{{\mathcal G}}
\newcommand{\F}{\mathcal{F}}

\newcommand{\Nat}{\mathbb{N}}
\newcommand{\bft}{{\bf true}\xspace}
\newcommand{\bff}{{\bf false}\xspace}
\newcommand{\buchi}{B\"uchi\xspace}
\newcommand{\zug}[1]{\langle #1 \rangle}

\newcommand{\pone}{\mbox{Player~$1$}\xspace}
\newcommand{\ptwo}{\mbox{Player~$2$}\xspace}
\newcommand{\pcov}{\mbox{Coverer}\xspace}
\newcommand{\pdis}{\mbox{Disruptor}\xspace}

\newcommand{\playeri}{Player~$i$\xspace}

\newcommand{\agenti}{Agent~$i$\xspace}

\newcommand{\npconp}{\mbox{$\Sigma_2^{\text{\sc P}}$}\xspace}
\newcommand{\outcome}{{\sf outcome}}
\newcommand{\succesor}{{\sf succ}}
\newcommand{\sat}{{\sf sat}}

\newcommand{\tuple}[1]{\langle #1 \rangle}
\newcommand{\set}[1]{\{ #1 \}}

\begin{document}

\title{Coverage Games}
\titlecomment{{\lsuper*}A preliminary version appeared
in the proceedings of CONCUR 2025, the 36th International Conference on Concurrency Theory.}
\thanks{Research supported by European Research Council, Advanced Grant ADVANSYNT}

\specialissue{Selected Papers of the 36th International Conference on Concurrency Theory (CONCUR 2025)}

\author[O.~Kupferman]{Orna Kupferman\lmcsorcid{0000-0003-4699-6117}}
\author[N.~Shenwald]{Noam Shenwald\lmcsorcid{0000-0003-1994-6835}}

\address{The Hebrew University, School of Computer Science and Engineering, Jerusalem, Israel}
\email{orna@cs.huji.ac.il, noam.shenwald@mail.huji.ac.il}

\begin{abstract}
  \noindent We introduce and study coverage games -- a novel framework for multi-agent planning in settings in which a system operates several agents but does not have full control over them, or interacts with an environment that consists of several agents.

The game is played between a coverer, who has a set of objectives, and a disruptor. The coverer operates several agents that interact with the adversarial disruptor.
The coverer wins if every objective is satisfied by at least one agent. Otherwise, the disruptor wins.

Coverage games thus extend traditional two-player games with multiple objectives by allowing a (possibly dynamic) decomposition of the objectives among the different agents. They have many applications, both in settings where the system is the coverer (e.g., multi-robot surveillance, coverage in multi-threaded systems) and settings where it is the disruptor (e.g., prevention of resource exhaustion, ensuring non-congestion).

We first study the theoretical properties of coverage games, including determinacy, and the ability to a priori decompose the objectives among the agents. We then study the problems of deciding whether the coverer or the disruptor wins. Besides a comprehensive analysis of the tight complexity of the problems, we consider interesting special cases, such as the one-player cases and settings with a fixed number of agents or objectives.
\end{abstract}

\maketitle
\section*{Introduction}

{\em Synthesis\/} is the automated construction of a system from its specification \cite{PR89a}. 
A {\em reactive\/} system interacts with its environment and has to satisfy its specification in all environments \cite{HP85}. 
A useful way to approach synthesis of reactive systems is to consider the situation as a {\em two-player game} between the system and the environment \cite{BCJ18}. 
The game is played on a graph whose vertices are partitioned between the players. Starting from an initial vertex, the players jointly move a token and generate a {\em play}, namely a path in the graph, with each player deciding the successor vertex when the token reaches a vertex she owns. The system wins if it has a strategy to ensure that no matter how the environment moves the token, the generated play satisfies an objective induced by the specification. 
For example, in {\em \buchi\/} games, the objective is given by a set $\alpha$ of vertices, and a play satisfies the objective if it visits $\alpha$ infinitely often.  

Synthesis is closely related to {\em automated planning\/} in AI \cite{TMA21}. Indeed, in  both settings, the goal is to realize a sequence of actions in order to accomplish a desired task, often under uncertainty or adversarial conditions. For example, in {\em pursuit-evasion} \cite{KA18}, a robot aims to complete a task while avoiding collisions with an adversarial robot. This setting naturally corresponds to a game between the robot and its environment, and extensive work from the synthesis community has been lifted to planning, for both finite- and infinite-horizon tasks \cite{KFP07,FJK10,ESK14,RPFK15,dGV15, JK15,LMFKKV16,CBMM18,MK20,DDTVZ21,GK21}. 

The game setting enables the modeling of uncertainty or adversarial environments, dealing with unknown terrains, partial observability, or limited sensing capabilities. 

Both synthesis and planning often involve {\em multiple objectives}. For example, in synthesis, generalized objectives such as {\em generalized \buchi} \cite{CDHL16}, {\em generalized parity} \cite{CHP07b}, {\em generalized reachability} \cite{FH10}, and {\em generalized reactivity} (GR(1)) \cite{PPS06,CDHL16} specify conjunctions of underlying objectives that the system must satisfy. Research on multi-objective games studies the impact of this conjunction on memory requirements, algorithmic complexity, and strategy construction. Such objectives have also been studied in a variety of game models, including concurrent, stochastic, energy, and weighted games \cite{BBMU12,CKK15,CDHR10,VCDHRR15,KS24}.

In planning, multiple objectives also naturally arise, especially in multi-agent systems \cite{TMA21}. The agents typically represent robots, drones, or autonomous vehicles that perform tasks such as terrain exploration or handling of requests at various locations \cite{Cho01,TMA21,CPLSK21,SB12,TPKR21}. For example, in {\em adversarial patrolling} \cite{LAK19}, a team of robots must detect intrusions into a guarded area, and in the {\em offline-coverage} problem \cite{AAK19}, robots must visit every point in a work area — typically within minimal time \cite{CKK15,TMA21,CPLSK21}.
In the context of planning from temporal specifications, researchers have studied settings in which several robots or components act cooperatively toward a shared objective \cite{MK20,GK21,ESK14}. For example, in \cite{MK20}, decentralized synthesis methods are used to coordinate a robot swarm so as to satisfy a global LTL specification, assuming full control and a static environment. Similarly, \cite{GK21} considers multi-robot synthesis from {\em Signal Temporal Logic} specifications using event-triggered control. Each robot is assigned its own task, and thus there is no strategic cooperation among the robots.

We introduce and study {\em coverage games} -- a framework for reasoning about planning tasks in which the system does not have full control over the agents. Formally, a coverage game is a two-player game with a set $\beta = \set{\alpha_1,\ldots,\alpha_m}$ of objectives, in which one player -- the covering player (named {\em \pcov}) operates a number $k$ of agents that together aim to satisfy all the objectives in $\beta$. In the beginning of the game, $k$ tokens are placed on the initial vertex. Each agent is responsible for one token, and as in standard two-player games, the interaction of each agent with the second player -- the disrupting player (named {\em \pdis}) generates a play in the graph. 
The {\em outcome} of the game is then a set of $k$ plays — one for each agent. The goal of \pcov is for these plays to cover all the objectives in $\beta$. That is, for all $1 \leq l \leq m$, there is an agent $1 \leq i \leq k$ such that $\alpha_l$ is satisfied in the play generated by the interaction of Agent~$i$ with \pdis. The game is {\em zero-sum}, thus the goal of \pdis is that at least one objective in $\beta$ is not satisfied in any of the plays. 

A {\em covering strategy} for \pcov is a vector of strategies, one for each agent, that ensures that no matter how \pdis behaves, each objective is satisfied in at least one play in the outcome. 
In the {\em coverage problem}, we are given a game graph $G$, a set $\beta$ of objectives, and a number $k \geq 1$ of agents, and we have to decide whether \pcov has a covering strategy. 

Note that when $k=1$, a coverage game coincides with a standard two-player game with multiple objectives. Also, it is not hard to see that when $k \geq m$, namely there are more agents than objectives, then a dominant strategy (in the game-theoretic sense) for \pcov allocates each objective in $\beta$ to a different agent. 
The interesting cases are when $1 < k < m$, in which case the objectives in $\beta$ need to be partitioned, possibly dynamically, among the agents. 

Coverage games significantly extend the settings addressed in current studies of planning. 
Indeed, in all examples discussed above, either there is a single agent, or the agents are under the full control of the designer (that is, no adversarial environment), or the specifications are pre-assigned to the different agents. Coverage games enable reasoning about multi-agent systems with a shared set of objectives, no a-priori assignment of objectives to agents, and no full control over their behavior.

For example, using coverage games with \buchi underlying objectives, one can reason about {\em multi-robot surveillance}, where we need to ensure that each set of critical locations is visited infinitely often in at least one robot's patrol route. There, vertices owned by \pdis correspond to locations in which we cannot control the movement of the robots, as well as positions in which the robots may encounter obstacles or experience a change in the landscape. Then, in {\em cyber-security systems}, the agents represent defense mechanisms, and the goal is to ensure that each set of potential attack vectors is mitigated infinitely often by at least one defense, countering an adversarial hacker. As another example, in {\em multi-threaded systems}, the agents are the processes, and we may want to ensure that each set of critical resources is accessed infinitely often, in all environments. 
Likewise, in {\em testing}, we want to activate and cover a set of functionalities of a software in all input sequences \cite{CWH21}.   

Unlike traditional two-player games, where the environment can be viewed as a system that aims to realize the negation of the specification, in coverage games \pdis's objective is not to cover the negation of the objectives in $\beta$, but rather to prevent \pcov from covering all the objectives in $\beta$. 
Formally, a {\em disrupting strategy} for \pdis is a strategy such that for all strategies of \pcov, there is at least one objective in $\beta$ that is not satisfied in each of the $k$ generated plays. Note that while \pcov activates several agents, \pdis follows a single strategy. 
This corresponds to environments that model the physical conditions of a landscape or servers whose response to users depends only on the interaction and not on the identity of the user. 
Then, in the {\em disruption problem}, we are given $G$, $\beta$, and $k$, and we have to decide whether \pdis has a disrupting strategy. 
While two-player games are determined, in the sense that in all games, one of the two players has a winning strategy, we are going to see that coverage games need not be determined, thus the coverage and disruption problems do not complement each other.

Viewing \pdis as the system and the agents as operated by the environment further motivates the study of coverage games. For example, in {\em intrusion detection\/} systems, the agents of the environment are potential intruders, and the system must ensure that at least one set of critical access points remains blocked. In {\em cloud computing}, the environment consists of client processes competing for resources, and the system's goal is to prevent resource exhaustion. Finally, in {\em traffic management}, the environment agents are the vehicles, and the system ensures that at least one route remains uncongested. 
Note that now, \pdis following a single strategy corresponds to systems like traffic controllers or vending machines, where the same policy is applied to all cars or customers that exhibit identical behavior.

A model related to coverage games from the point of view of the disruptor is 
{\em population games\/} \cite{BDGGG18}. These games focus on controlling a homogeneous population of agents that have the same non-deterministic behavior (see also \cite{CFO21,GMT24}). The edges of a two-player population game graph are labeled by actions. The game is nondeterministic, thus different edges that leave the same vertex may be labeled by the same action. As in the disruption problem, each agent moves a token along the graph. In each turn, a controller chooses the same action for all agents, and each agent decides how to resolve nondeterminism and proceeds to a successor vertex along an edge labeled with the action. 
The goal of the controller is to synchronize all the agents to reach a target state. 

The technical details of population games are very different from those of coverage games. Indeed, there, the controller takes the same action in all of its interactions with the agents, and the type of objectives are different from the coverage and disruption objectives in coverage games.  

We study the coverage and disruption problems in coverage games with {\em \buchi} and {\em co-\buchi} objectives. We first study the theoretical aspects of coverage games. We warm up with games in which the number of agents enables an easy reduction to usual two-player games (in particular, when $k \geq |\beta|$, the setting corresponds to $|\beta|$ two-player games with a single objective), and continue to {\em one-player coverage games}, where all vertices are owned by one player. 
In particular, when all vertices are owned by \pcov, the coverage problem boils down to finding $k$ paths that cover all the objectives in $\beta$, and our complexity results are aligned with the NP-hardness of different variants of the {\em coverage path planning} problem for robot swarms \cite{AAK19}.

We then continue to the general setting, show that it is undetermined (that is, possibly neither \pcov nor \pdis has a coverage or disruption strategy), and examine the problem of an {\em a-priori decomposition\/} of the objectives in $\beta$ among the agents. 
We show that the objectives cannot be decomposed in advance, and characterize covering strategies for \pcov as ones in which all agents satisfy all objectives as long as such a decomposition is impossible. The characterization is the key to our upper bounds for the complexity of the coverage and disruption problems. 

The coverage and disruption problems have three parameters: the game graph $G$, the set $\beta$ of objectives, and the number $k$ of agents. 
In typical applications, the game graph $G$ is much bigger than $k$ and $\beta$. It is thus interesting to examine the complexity of CGs in terms of the size of $G$, assuming the other parameters are fixed. We provide a comprehensive {\em parameterized-complexity\/} analysis of the problems. Note that fixing $G$ also fixes $\beta$. Also, as cases with $k \geq |\beta|$ are easy, fixing $\beta$ essentially fixes $k$. 
Accordingly, we study settings in which only $k$ or $|\beta|$ are fixed. 
Our results are summarized in Table~\ref{table results}. All the complexities in the table are tight.

\begin{table}[htb]
\centering
\begin{tabular}{c ccc ccc}
\toprule
 & \multicolumn{3}{c}{B\"uchi} & \multicolumn{3}{c}{co-B\"uchi} \\
\cmidrule(lr){2-4} \cmidrule(lr){5-7}
fixed: & - & $k$ & $|\beta|$ & - & $k$ & $|\beta|$ \\
\midrule
Coverage
 & PSPACE & NP & PTIME
 & PSPACE & NP & PTIME \\
 & {\footnotesize Ths.~\ref{decide cover buchi upper},~\ref{decide cover lower}}
 & {\footnotesize Th.~\ref{2agent cover buchi np hard}}
 & {\footnotesize Th.~\ref{fixed beta coverage}}
 & {\footnotesize Ths.~\ref{decide cover buchi upper},~\ref{decide cover lower}}
 & {\footnotesize Th.~\ref{2agent cover buchi np hard}}
 & {\footnotesize Th.~\ref{fixed beta coverage}} \\
\midrule
Disruption
 & $\npconp$ & NP & PTIME
 & $\npconp$ & $\npconp$ & PTIME \\
 & {\footnotesize Ths.~\ref{decide disrupt buchi upper},~\ref{decide disrupt lower}}
 & {\footnotesize Th.~\ref{2agent disrupt buchi np-c}}
 & {\footnotesize Th.~\ref{fixed beta disruption}}
 & {\footnotesize Th.~\ref{decide disrupt CGC upper},~\ref{decide disrupt lower}}
 & {\footnotesize Th.~\ref{disrupt co-buchi fixed k}}
 & {\footnotesize Th.~\ref{fixed beta disruption}} \\
\bottomrule
\end{tabular}
\caption{The complexity of the coverage and disruption problems. All complexities are tight.}
\label{table results}
\end{table}

\noindent
For both Büchi and co-Büchi objectives, the general problem is PSPACE-complete for coverage and $\npconp$-complete for disruption, in contrast to the PTIME complexity of standard two-player games with AllB (and the dual ExistsC) and AllC (and the dual ExistsB) objectives. For coverage, fixing the number $k$ of agents reduces the complexity to NP, and fixing the number $|\beta|$ of objectives results in PTIME complexity. Essentially, a fixed $k$ bounds the recursion depth of the general PSPACE algorithm (Theorem~\ref{decide cover buchi upper}) by a constant, whereas a fixed $|\beta|$ enables the elimination of the nondeterminism in the algorithm. 

Note that for the coverage problem, the complexities for games with \buchi and co-\buchi objectives coincide. For the disruption problem, the picture is different. While the joint complexity as well as the complexity for a fixed $|\beta|$ coincide for \buchi and co-\buchi, fixing $k$ only leaves the complexity $\npconp$ for co-\buchi objectives and takes it down to NP in the \buchi case. This is due to the NP-hardness of one-player coverage games with underlying co-\buchi objectives, which already holds for a fixed number of agents.  

From a technical point of view, for the upper bounds, the main challenging results are the characterization of the decomposability of objectives and the recursive algorithm it entails (Theorems~\ref{covering str characterization} and~\ref{decide cover buchi upper}), as well as the ability to work with a short symbolic representation of strategies that are not memoryless (Lemma~\ref{disrupting iff fair}). For the lower bounds, the main technical issue is utilizing the agents to model satisfying assignments in Boolean satisfaction problems, and the most technically-involved results are for the cases with a fixed number of agents (Theorems~\ref{2agent disrupt buchi np-c} and ~\ref{disrupt co-buchi fixed k}). In particular, for co-\buchi objectives, moving from a single agent to two agents increases the complexity from PTIME to $\npconp$.

We discuss possible variants and extensions of coverage games. Beyond classical extensions of the underlying two-player games (e.g., concurrency, stochastic settings, partial visibility, etc.) and to the objectives (e.g., richer winning conditions, weighted objectives, etc.), as well as classical extensions from planning (optimality of agents and their resources), we focus on elements that have to do with the operation of several agents: communication among the agents (in our setting, the strategies of the agents are independent of each other), and the ability of \pdis to also use different agents.  
 
\section{Preliminaries}

\subsection{Two-player games}
A \emph{two-player game graph} is a tuple $G = \tuple{V_1,V_2, v_0, E}$, where $V_1$, $V_2$ are disjoint sets of vertices, owned by \pone and \ptwo, respectively, and we let $V=V_1\cup V_2$. Then, $v_0 \in V$ is an initial vertex and $E \subseteq\  V\times V$ is a total edge relation, thus for every $v\in V$, there is $u\in V$ such that $\zug{v,u}\in E$. The size of $G$, denoted $|G|$, is $|E|$, namely the number of edges in it. When we draw game graphs, the vertices in $V_1$ and $V_2$ are drawn as circles and squares, respectively.

At the beginning of a play in the game, a token is placed on $v_0$. Then, in each turn, the player that owns the vertex that hosts the token chooses a successor vertex and moves the token to it. Together, the players generate a \emph{play} $\rho=v_0,v_1,\ldots$ in $G$, namely an infinite path that starts in $v_0$ and respects $E$: for all $i \geq 0$, we have that $\tuple{v_i,v_{i+1}}\in E$. 

A {\em strategy} for \playeri is a function $f_i:V^*\cdot V_i\rightarrow V$ that maps prefixes of plays that end in a vertex owned by \playeri to possible extensions of the play. That is, for every $\rho \in V^*$ and $v \in V_i$, we have that $\tuple{v,f_i(\rho\cdot v)}\in E$. Intuitively, a strategy for \playeri directs her how to move the token, and the direction may depend on the history of the game so far. 
A strategy is {\em memoryless\/} if its choices depend only on the current vertex and are independent of the history of the play. Accordingly, we describe a memoryless strategy for \playeri by a function $f_i: V_i\rightarrow V$. 

A \textit{profile} is a tuple $\pi=\tuple{f_1,f_2}$ of strategies, one for each player. The \emph{outcome} of a profile $\pi=\tuple{f_1,f_2}$ is the play obtained when the players follow their strategies in $\pi$. Formally, $\outcome(\pi) = v_0, v_1, \ldots \in V^\omega$ is such that for all $ j \geq 0$, we have that $v_{j+1}=f_i(v_0,v_1,\ldots,v_j)$, where $i\in\set{1,2}$ is such that $v_j\in V_i$. 

A {\em two-player game\/} is a pair $\G=\zug{G,\psi}$, where $G= \tuple{V_1,V_2, v_0, E}$ is a two-player game graph, and $\psi$ is an objective for \pone, specifying a subset of $V^\omega$, namely the set of plays in which \pone wins. We discuss different types of objectives below. The game is zero-sum, thus \ptwo wins when the play does not satisfy $\psi$. A strategy $f_1$ is a {\em winning strategy} for \pone if for every strategy $f_2$ for \ptwo, we have that \pone wins in $\tuple{f_1,f_2}$, thus $\outcome(\zug{f_1,f_2})$ satisfies $\psi$. Dually, a strategy $f_2$ for \ptwo is a winning strategy for \ptwo if for every strategy $f_1$ for \pone, we have that \ptwo wins in $\tuple{f_1,f_2}$. 
We say that \playeri $ $\/{\em wins} in $\G$ if she has a winning strategy. A game is {\em determined} if \pone or \ptwo wins it. 

For a play $\rho=v_0,v_1,\ldots$, we denote by 
$\textit{inf}(\rho)$ the set of vertices that are visited infinitely often along $\rho$. That is, 
$\textit{inf}(\rho)=\set{v\in V : \text{there are infinitely many } i\geq 0 \text{ such that } v_i = v}$. For a set of vertices $\alpha\subseteq V$, a play $\rho$ satisfies the {\em \buchi} objective $\alpha$ iff $\textit{inf}(\rho)\cap \alpha\neq \emptyset$. The objective dual to 
\buchi is {\em co-\buchi}. 
Formally, a play $\rho$ satisfies 
a co-\buchi objective $\alpha$ iff $\textit{inf}(\rho)\cap\alpha = \emptyset$. 
We use $\gamma \in \set{\text{B,C}}$ to denote the different objective types. 

For a play $\rho$ and a set of objectives $\beta = \set{\alpha_1,\ldots,\alpha_m}$, we denote by $\sat(\rho, \beta)$ the set of objectives $\alpha_i \in \beta$ that are satisfied in $\rho$. An {\em All objective} is a set $\beta$ of objectives, all of the same type. A play $\rho$ satisfies an All objective $\beta$ iff $\rho$ satisfies every objective in $\beta$. That is, if $\sat(\rho, \beta) = \beta$. The objective dual to All is {\em Exists}.\footnote{All-\buchi objectives are traditionally referred to as {\em generalized \buchi}. Then, their dual Exists co-\buchi objectives are referred to as {\em generalized co-\buchi}
\cite{KPV06}. We opt to follow a terminology with All and Exists, making the quantification within the satisfaction requirement clearer.} Formally, a play $\rho$ satisfies an Exists objective $\beta$ iff $\rho$ satisfies at least one objective in $\beta$. That is, if $\sat(\rho, \beta) \cap \beta \neq \emptyset$. Note that an All-co-\buchi objective $\beta$ is equal to the co-\buchi objective $\cup \beta$.

\subsection{Two-player multi-agent coverage games}

In coverage games, \pone operates a number of agents. All agents play against \ptwo, and so the outcome of the game is a set of plays -- one for each agent. \pone has multiple objectives, and her goal is to cover all the objectives, in the sense that each objective is satisfied in at least one of the plays in the outcome. Accordingly, we refer to \pone 
as the covering player, named {\em \pcov}, and refer to \ptwo as the disrupting player, named {\em \pdis}. 

Formally, for $\gamma \in \set{\text{B,C}}$, a {\em two-player multi-agent $\gamma$-coverage game} ($\gamma$-CG, for short) is a tuple $\G = \tuple{G, k, \beta}$, where $G$ is a two-player game graph, $k\in \Nat$ is the number of agents \pcov operates, and $\beta = \set{\alpha_1,\ldots,\alpha_m}$ is a set of objectives of type $\gamma$. When $\gamma$ is not important or is clear from the context, we omit it and refer to $\G$ as a CG. 

For $k \geq 1$, let $[k]=\{1,\ldots,k\}$. A {\em strategy} for \pcov is a tuple $F_1 = \tuple{f_1^1,\ldots,f_1^k}$ of $k$ strategies for \pcov in the game graph $G$. We assume that \pdis uses the same strategy $f_2$ against all the agents of \pcov. 
A {\em profile} is a tuple $\pi = \tuple{F_1,f_2}$ of strategies $F_1 = \tuple{f_1^1,\ldots,f_1^k}$ for \pcov and $f_2$ for \pdis. The outcome of $\pi$ is the tuple of the $k$ plays generated when \pcov's agents and \pdis follow their strategies. Formally, $\outcome(\pi) = \tuple{\rho_1,\ldots,\rho_k}$, where $\rho_i = \outcome(\tuple{f_1^i,f_2})$, for every $i\in [k]$.   

We say that a profile $\pi$ {\em covers\/} $\beta$ iff every objective in $\beta$ is satisfied in at least one play in $\outcome(\pi)$. That is, if $\cup\set{\sat(\rho,\beta): \rho \in\outcome(\pi)} = \beta$. Equivalently, for every $j\in [m]$, there exists $i\in [k]$ such that $\alpha_j$ is satisfied in $\outcome(\tuple{f_1^i,f_2})$. A strategy $F_1$ for \pcov is a {\em covering} strategy in $\G$ iff for every strategy $f_2$ for \pdis, the profile $\tuple{F_1,f_2}$ covers $\beta$. 
Then, we say that \pcov wins $\G$ iff she has a covering strategy in $\G$. A strategy $f_2$ for \pdis is a {\em disrupting} strategy in $\G$ iff for every strategy $F_1$ for \pcov, the profile $\tuple{F_1,f_2}$ does not cover $\beta$. Then, we say that \pdis wins $\G$ iff she has a disrupting strategy in $\G$. 
Finally, we say that a $\gamma$-CG $\G$ is {\em determined\/} if \pcov has a covering strategy in $\G$ or \pdis has a disrupting strategy in $\G$. 

In the {\em coverage problem for CGs}, we have to decide whether \pcov has a covering strategy in a given CG. In the {\em disruption problem for CGs}, we have to decide whether \pdis has a disrupting strategy in a given CG.

\begin{exa}
\label{example}
Consider the \buchi CG $\G = \tuple{G,2,\beta}$, with $\beta = \set{\alpha_1,\alpha_2,\alpha_3}$, for the game graph $G$ appearing in Figure~\ref{fig G avoid}, and $\alpha_1 = \set{u_1,u_2,u_3}$ (red), $\alpha_2 = \set{m_1,u_2,d_3}$ (green), and 
$\alpha_3 = \set{d_1,d_2,d_3}$ (yellow). Thus, \pcov should direct two agents in a way that ensures that the three colors are visited infinitely often. 

  \begin{figure}[htb]
    \centering
    \includegraphics[width=7cm]{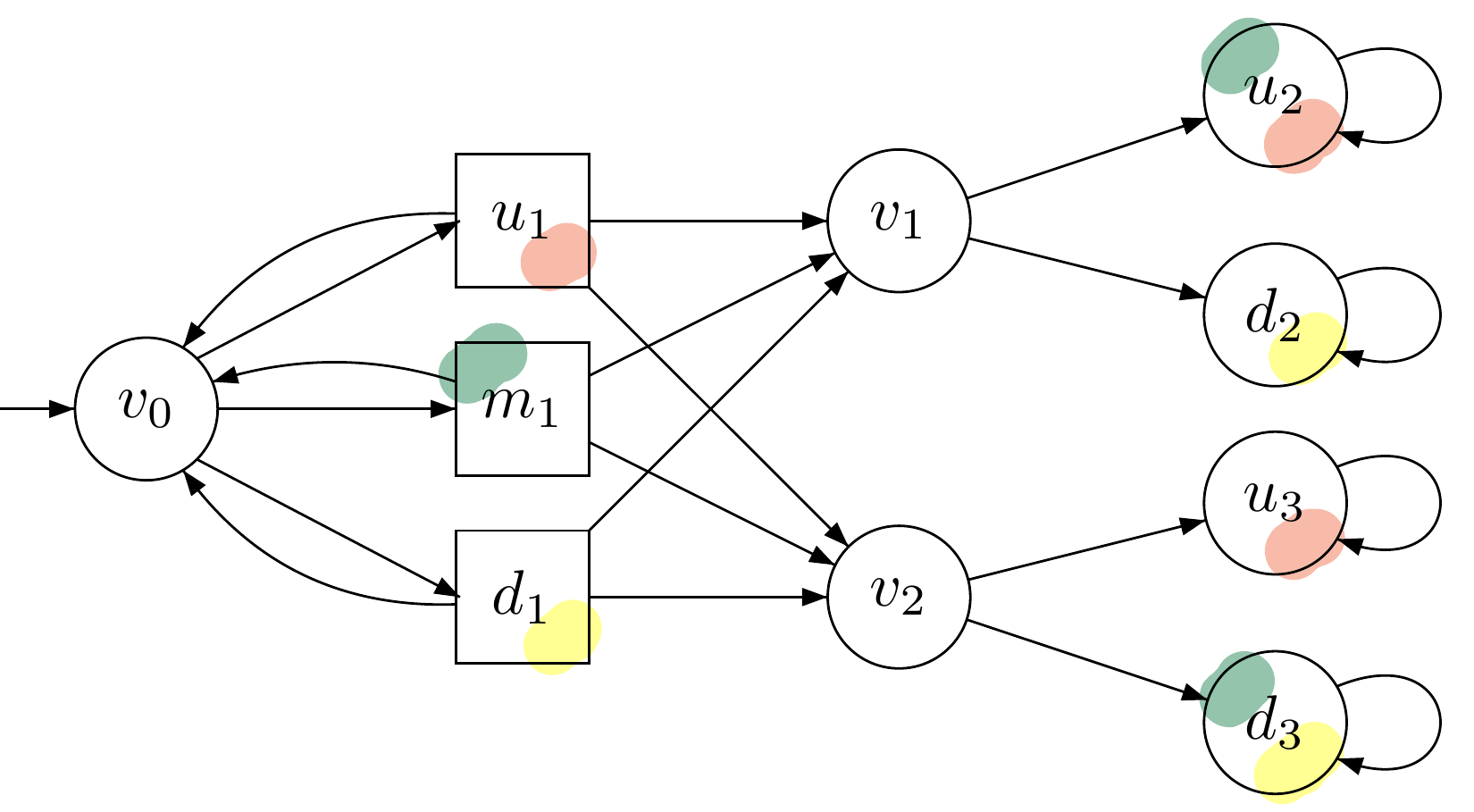}
    \caption{The \buchi CG $\G$.}
    \label{fig G avoid}
\end{figure}   

\noindent
\pcov has a covering strategy in $\G$: In $v_0$, the tokens of both agents move together, visiting the three colors in a round-robin fashion. Specifically, as long as \pdis directs the tokens back to $v_0$, \pcov can visit $u_1$, $m_1$, and $d_1$ sequentially.  
If the play stays forever in the left sub-game, $\beta$ is covered. If \pdis decides to leave the left sub-game and directs the tokens to $v_1$ or $v_2$, the agents split $\beta$ between them. One agent covers two colors (green and red in $u_2$, in case the tokens are in $v_1$, or green and yellow in $d_3$, in case the tokens are in $v_2$), and the second agent covers one color (yellow or red, respectively). Thus, $\beta$ is covered in this case too. \hfill \qed
\end{exa}

\section{Special Cases of Coverage Games}
\label{sec easy}
As a starter, in this section we examine two special cases of CGs: games in which the number of agents enables an easy reduction to traditional two-player games, and one-player CGs, thus when one player owns all vertices.

\subsection{CGs with one or many agents}
\label{k easy}
Consider a CG $\G = \tuple{G, k,\beta}$.
When $k=1$, the CG $\G$ is equivalent to the two-player game with the All objective $\beta$. Indeed, the single play induced by the strategy of the single agent has to cover all the objectives in $\beta$. 
Thus, the case $k=1$ can be solved in PTIME, serving as a lower bound for the general case, or more precisely, as a reference point to the study of the complexity for the general case. 

When $k\geq |\beta|$, there is no need to assign more than one objective in $\beta$ to each of the agents. Formally, we have the following.   

\begin{lem}
\label{P1 k geq m wins iff she m wins}
Consider a $\gamma$-CG $\G = \tuple{G, k,\beta}$ with $\gamma\in\set{\text{B,C}}$ and $k\geq |\beta|$, and the following statements. 
\begin{description}
\item[(C1)]
\pcov has a covering strategy in $\G$.
\item[(C2)]
For every $\alpha_i \in \beta$, \pone wins the two-player $\gamma$-game $\G_i = \tuple{G,\alpha_i}$. 
\item[(C3)]
There is $\alpha_i \in \beta$ such that \ptwo wins the two-player $\gamma$-game $\G_i = \tuple{G,\alpha_i}$.
\item[(C4)]
\pdis has a disrupting strategy in $\G$.
\end{description}
Then, $(C1)$ iff $(C2)$, $(C3)$ iff $(C4)$, and $(C1)$ and $(C4)$ complement each other. 
\end{lem}
\begin{proof}
First note that since two-player $\gamma$-games are determined, then $(C2)$ and $(C3)$ complement each other. Also, as it is impossible for both \pcov to have a covering strategy and \pdis to have a disrupting strategy, $(C1)$ and $(C4)$ contradict each other. Hence, it is sufficient to prove that $(C2)$ implies $(C1)$ and that $(C3)$ implies $(C4)$.

Assume that $(C2)$ holds. For every $\alpha_i \in \beta$, let $f_1^i$ be a winning strategy for \pcov in $\G_i = \tuple{G,\alpha_i}$. It is easy to see that $F_1 = \tuple{f_1^1,\ldots,f_1^{|\beta|}, f_1^{|\beta|+1},\ldots,f_1^{k}}$ is a covering strategy for \pcov, for arbitrary strategies $f_1^{|\beta|+1},\ldots,f_1^{k}$ for agents $|\beta|+1,\ldots,k$. Indeed, the strategy $F_1$ guarantees that \agenti, for $1 \leq i \leq |\beta|$, satisfies $\alpha_i$.  Thus, \pcov has a covering strategy, and $(C1)$ holds. 

Assume now that $(C3)$ holds. Thus, there is $\alpha_i \in \beta$ such that \pdis wins the two-player $\gamma$-game $\G_i = \tuple{G,\alpha_i}$. Let $f_2$ be the winning strategy for \pdis in $\G_i$. Then, it must be that $\alpha_i\notin \sat(\rho, \beta)$ for every strategy $F_1$ for \pcov and every play $\rho \in \outcome(\tuple{F_1,f_2})$. Thus, \pdis has a disrupting strategy and $(C4)$ holds. 
\end{proof}

\noindent
By Lemma~\ref{P1 k geq m wins iff she m wins}, deciding CGs with $k\geq |\beta|$ agents can be reduced to deciding $|\beta|$ two-player games, thus such CGs are determined and can be decided in PTIME \cite{Mar75,VW86a,Zie98}.

When $1 < k < |\beta|$, \pcov has to partition the objectives in $\beta$ among the different agents. As we study in Section~\ref{prop cg}, this makes the setting much more challenging and interesting. 

\subsection{One-player CGs}

We proceed to one-player CGs, thus when one of the players owns all the vertices. 
Formally, $\G = \tuple{G, k, \beta}$, with $G=\zug{V_1,V_2,v_0,E}$ is such that $V_2 =\emptyset$ (equivalently, $V=V_1$) or $V_1 = \emptyset$ (equivalently, $V=V_2$). 
Note that one-player CGs are determined, in the sense that whenever \pcov does not have a covering strategy, \pdis has an (empty) disrupting strategy, and vice versa. 

We study the complexity of the coverage and disruption problems in one-player CGs. Our results are summarized in Table~\ref{one player comp}.

\begin{table}[htb]
\centering
\begin{tabular}{c ccc ccc}
\toprule
 & \multicolumn{3}{c}{B\"uchi} & \multicolumn{3}{c}{co-B\"uchi} \\
\cmidrule(lr){2-4} \cmidrule(lr){5-7}
fixed: & - & $k$ & $|\beta|$ & - & $k$ & $|\beta|$ \\
\midrule
$V=V_1$
 & NP & NLOGSPACE & NLOGSPACE
 & NP & NP & NLOGSPACE \\
 & {\footnotesize Th.~\ref{V2 = emptyset}}
 & {\footnotesize Th.~\ref{V2 = emptyset and fixed k B}}
 & {\footnotesize Th.~\ref{V2 = empty fixed beta}}
 & {\footnotesize Th.~\ref{V2 = emptyset}}
 & {\footnotesize Th.~\ref{V2 = emptyset and fixed k C}}
 & {\footnotesize Th.~\ref{V2 = empty fixed beta}} \\
\midrule
$V=V_2$
 & \multicolumn{6}{c}{NLOGSPACE} \\
 & \multicolumn{6}{c}{{\footnotesize Th.~\ref{V1 = emptyset}}} \\
\bottomrule
\end{tabular}
\caption{The complexity of deciding one-player CGs. All complexities are tight.}
\label{one player comp}
\end{table}

\noindent
Consider first the case \pcov has full control. Thus, $G=\zug{V,\emptyset,v_0,E}$. Since \pcov has full control, covering $\beta$ amounts to finding $k$ paths in $G$ such that each objective $\alpha_i \in \beta$ is satisfied in at least one path. A path in $G$ is {\em lasso-shaped} if it is of the form $p\cdot q^\omega$, with $p\in V^*$ and $q\in V^+$. The length of $p\cdot q^\omega$ is defined as $|p|+|q|$.

\begin{thm}
\label{V2 = emptyset}
The coverage problem for \buchi or co-\buchi CGs with $V = V_1$ is NP-complete. 
\end{thm}
\begin{proof}
We start with the upper bounds. Consider a CG $\G = \tuple{G,k,\beta}$ with $V=V_1$. 
An NP algorithm guesses $k$ lasso-shaped paths in $G$ of length of most $|V|\cdot |\beta|$, and checks that every objective $\alpha_i \in \beta$ is satisfied in at least one of them. 

Clearly, if such $k$ paths exist, then \pcov has a covering strategy. Also, checking lasso-shaped paths of length $|V|\cdot |\beta|$ is sufficient. 
Indeed, for both $\gamma\in \set{\text{B,C}}$, if there exists a path in $G$ that satisfies a subset $\beta'\subseteq \beta$ of $\gamma$ objectives, then there also exists a cycle in $G$ of length at most $|V|\cdot|\beta'|$ that visits or avoids all the sets in $\beta'$. Since $|\beta'|\leq |\beta|$, the bound follows. 

For the lower bound, we describe (easy) reductions from the vertex-cover problem. Consider an undirected graph $G = \tuple{V,E}$, and $k\leq |V|$. Recall that a vertex cover for $G$ is a set $U \subseteq V$ such that for all edges $\set{u,v} \in E$, we have that $\set{u,v} \cap U \neq \emptyset$. For both $\gamma\in \set{\text{B,C}}$, we construct a $\gamma$-CG $\G = \tuple{G',k,\beta}$ with $V=V_1$ such that \pcov has a covering strategy in $\G$ iff there exists a vertex cover of size at most $k$ in $G$. 

The game graph $G'$ is independent of $E$ and consists of the vertices in $V$ and a new initial vertex $v_0$. The only edges are self-loops in the vertices in $V$, and edges from $v_0$ to all vertices in $V$. Accordingly, the $k$ agents of \pcov essentially choose $k$ vertices in $V$. The objectives in $\beta$ then require these $k$ vertices to be a vertex cover. For $\gamma = \text{B}$, we define $\beta = E$. That is, for every edge $e = \set{u,v}\in E$, the objective $e$ is to visit one of the vertices that correspond to $u$ or $v$ infinitely often. Accordingly, a covering strategy for \pcov induces a vertex cover, and vice versa.
For $\gamma = \text{C}$, we define $\beta = \set{V \setminus e : e\in E}$. Since each play $\rho$ is eventually trapped in a self-loop in a vertex in $V$, the obtained game coincides with the one for $\gamma = \text{B}$.

We describe the reductions formally and prove their correctness. 
We start with the case $\gamma = \text{B}$. There, $\G = \tuple{G',k,E}$, where the game graph $G' = \tuple{V',\emptyset,v_0,E'}$ has the following components.
    \begin{enumerate}
        \item $V' = V\cup \set{v_0}$, for $v_0 \notin V$.
        \item The set $E'$ of edges contains, for every $v\in V$, the edges $\tuple{v_0,v}$ and $\tuple{v,v}$. That is, from the initial vertex, \pcov chooses a vertex $v\in V$ by proceeding to the vertex that corresponds to $v$ in $V'$, and then the play loops forever in $v$.  
     \end{enumerate}
\noindent
We prove the correctness of the reduction. Assume first that there exists a vertex-cover $U\subseteq V$ of size $k$ in $G$. Then, the strategy $F_1$ for \pcov in $\G$ that sends a different agent to every vertex $v\in U$ is a covering strategy. Indeed, since $U$ is a vertex cover, for every edge $e=\set{u,v}\in E$ we have that $u\in U$ or $v\in U$, thus the \buchi objective $e$ is satisfied in the play in $\outcome(F_1)$ that reaches the self-looped sink $v$ or the self-looped sink $u$.

For the second direction, assume there exists a covering strategy $F_1$ in $\G$. Let $U$ be the set of $k$ self-looped sinks that the plays in $\outcome(F_1)$ reach. Then, $U$ is a vertex-cover in $G$ of size $k$. Indeed, for every edge $e = \set{u,v}\in E$, if every play in $\outcome(F_1)$ does not reach $u$ and does not reach $v$, then there does not exist a play in $\outcome(F_1)$ that satisfies the \buchi objective $e$, which is not possible since $F_1$ is a covering strategy.

For $\gamma = \text{C}$, we define the same graph $G$, with $\beta = \set{V \setminus e : e\in E}$. Since each play $\rho$ is eventually trapped in a self-loop in a vertex in $V$, the co-\buchi objective $V\setminus e$ is satisfied in $\rho$ iff the \buchi objective $e$ is satisfied in $\rho$, thus the game coincides with the one for $\gamma = \text{B}$.
\end{proof}
\noindent
When \pcov owns all the vertices and the number of agents is fixed, the problem is easier for $\gamma = B$, but retains its NP-hardness for $\gamma=C$. 

Intuitively, the difference in the complexities has to do with the difficulty of finding the maximal sets of objectives that one agent can satisfy on her own. In \buchi CGs, finding these maximal sets is easy, and thus when the number of agents is fixed, finding a combination of them that covers $\beta$ can be done in NLOGSPACE. On the other hand, in co-\buchi CGs, the hardness of the problem stems from the hardness of finding those maximal sets, and so the problem is NP-hard already for a fixed number of agents.  
Formally, we have the following. 

\begin{thm}
\label{V2 = emptyset and fixed k B}
The coverage problem for \buchi CGs with $V=V_1$ and a fixed number of agents is NLOGSPACE-complete.
\end{thm}
\begin{proof}
Consider a CG $\tuple{G,k,\beta}$ with $V=V_1$. We first argue that \pcov has a covering strategy iff $G$ has $k$ reachable, non-trivial, and strongly-connected components (SCCs) $C_1,\ldots,C_k$ such that $\bigcup_{i\in[k]}\set{\alpha_j\in \beta: \alpha_j \cap C_i \neq \emptyset} = \beta$.
Indeed, each outcome of the game can reach and stay in an SCC $C$ of $G$, where it can traverse infinitely often all the vertices in $C$. 
Since $k$ is fixed, an NLOGSPACE algorithm can thus guess $k$ non-trivial SCCs, and check that every objective $\alpha \in \beta$ intersects with one of the SCCs. 

In more detail, the NLOGSPACE algorithm guesses $k$ vertices $v_1,\ldots,v_k \in V$, checks that they are all reachable from $v_0$, and writes them on the tape. Since $k$ is fixed, the space required is logarithmic in $|V|$.
Then, for each objective $\alpha \in \beta$, the algorithm guesses $1\leq i\leq k$ and checks whether there exists a path from $v_i$ to $\alpha$ and back to $v_i$, which can be done in NLOGSPACE. If such a path does not exist, the algorithm rejects. If the algorithm successfully completes the check for all the objectives in $\beta$, it accepts.
 
Hardness in NLOGSPACE follows from an easy reduction from deciding one-player \buchi games, a.k.a. nonemptiness of B\"uchi word automata \cite{VW94}.
\end{proof}

\noindent
While fixing the number of agents reduces the complexity of the coverage problem for $\gamma = B$, it does not have the same effect for $\gamma = C$. This marks the first difference between the two settings. Formally, we have the following.

\begin{thm}
\label{V2 = emptyset and fixed k C}
The coverage problem for co-\buchi CGs with $V=V_1$ is NP-hard, already for CGs with two agents. 
\end{thm}
\begin{proof}
We describe a reduction from $3$SAT. Consider a set of variables $X = \set{x_1,\ldots,x_n}$, and let $\overline{X} = \set{\overline{x_1},\ldots,\overline{x_n}}$. Also, let $\phi = C_1 \wedge \cdots \wedge C_m$ be a propositional Boolean formula given in $3$CNF. That is, $C_i = (l_i^1\vee l_i^2\vee l_i^3)$ with $l_i^1,l_i^2,l_i^3\in X\cup \overline{X}$, for every $i\in [m]$. We construct from $\phi$ a co-\buchi CG $\G = \tuple{G,2,\beta}$, such that $V=V_1$ and \pcov has a covering strategy in $\G$ iff $\phi$ is satisfiable.

The game graph $G$ is defined as follows (see Figure~\ref{fig sat 2 agents V=V1}). From the initial vertex, \pcov chooses between proceeding to an {\em assignment sub-graph}, and to a {\em clause sub-graph}. In the assignment sub-graph, \pcov chooses an assignment to the variables in $X$, where choosing an assignment to a variable $x$ involves choosing between a vertex that corresponds to the literal $x$, and a vertex that corresponds to the literal $\overline{x}$. \pcov then repeats this process infinitely often.

\begin{figure}[htp]
\centering
\includegraphics[width=15cm]{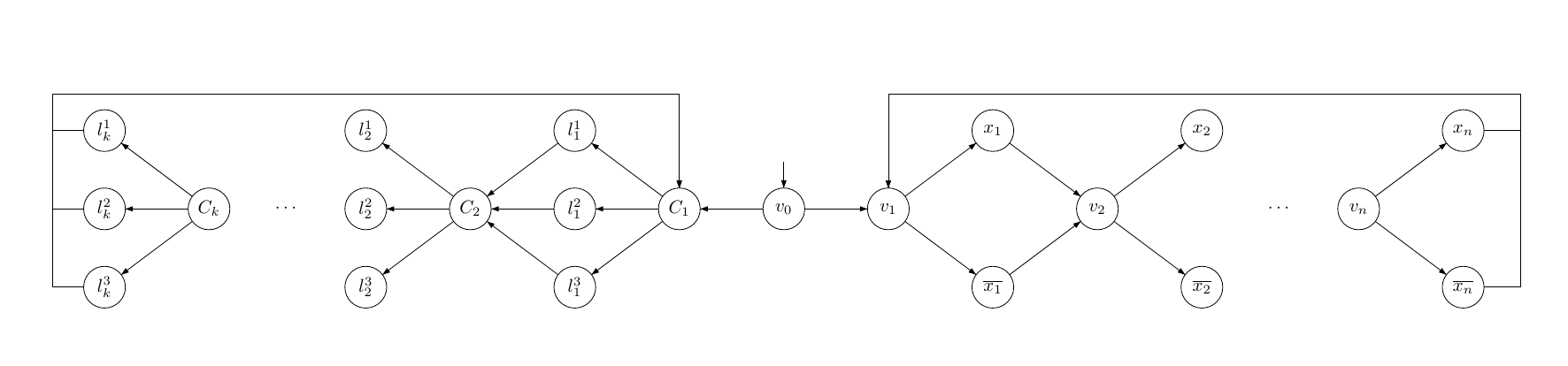}
\caption{The game graph $G$.}
\label{fig sat 2 agents V=V1}
\end{figure}

\noindent
In the clause sub-graph, \pcov chooses a literal for every clause of $\phi$, which involves choosing between vertices that correspond to the literals $l_i^1, l_i^2$ and $l_i^3$, for every $i\in [m]$. As in the assignment sub-graph, \pcov then repeats this process. 

\pcov has to cover the following objectives. The first two objectives $\alpha_1$ and $\alpha_2$ are the vertices of the assignment sub-graph and the vertices of the clause sub-graph, respectively. Note that in order to satisfy both objectives, \pcov must send one agent to the assignment sub-graph, and one agent to the clause sub-graph. We call those agents the {\em assignment agent} and the {\em clause agent}, respectively. Then, for every literal $l\in X\cup \overline{X}$, we define the objective $\alpha_l$ as the set that consists of the vertex that corresponds to $l$ in the assignment sub-graph, and every vertex that corresponds to $\overline{l}$ in the clause sub-graph. 

Accordingly, the objective $\alpha_l$ is satisfied iff the assignment agent visits the vertex that corresponds to $l$ only finitely often, or the clause agent visits vertices that correspond to $\overline{l}$ only finitely often.

Intuitively, if the assignment agent chooses a satisfying assignment and the clause agent chooses literals that are evaluated to $\bft$ in the chosen assignment, then every objective $\alpha_l$ is satisfied. Hence, a satisfying assignment to $\phi$ induces a covering strategy for \pcov.
Also, if $\phi$ is not satisfiable, then for every assignment induced by the assignment agent, there is a clause in $\phi$ that is not satisfied. Then, the clause agent causes the objective of at least one of the literals of this clause to be not satisfied.  

Formally, we define $\G = \tuple{G, 2, \beta}$ as follows.
First, the game graph $G = \tuple{V_1,\emptyset,v_0,E}$ contains the following components (see Figure~\ref{fig sat 2 agents V=V1}).
    \begin{enumerate}
        \item $V_1 = \set{v_0} \cup \set{v_i: i\in [n]} \cup X\cup \overline{X} \cup \set{C_i : i\in [m]} \cup \set{l_i^1,l_i^2,l_i^3 : i\in [m]}$. The vertices in $\set{v_i: i\in [n]}$ are {\em variable vertices}, and the vertices in $\set{C_i : i\in [m]}$ are {\em clause vertices}.

        Also, the vertices in $X\cup \overline{X}$ are {\em literal vertices}, and the vertices in $\set{l_i^1,l_i^2,l_i^3 : i\in [m]}$ are {\em clause-literal vertices}.
        \item The set of edges $E$ contains the following edges.
        \begin{enumerate}
            \item $\tuple{v_0,v_1}$.
            \item $\tuple{v_i, x_i}$ and $\tuple{v_i,\overline{x_i}}$, for every $i\in [n]$.
            \item $\tuple{x_i,v_{(i+1)\mod n}}$ and $\tuple{\overline{x_i},v_{(i+1)\mod n}}$, for every $i\in [n]$.
            \item $\tuple{v_0,C_1}$.
            \item $\tuple{C_i,l_i^j}$, for every $i\in [m]$ and $j\in \set{1,2,3}$.
            \item $\tuple{l_i^j,C_{(i+1)\mod m}}$, for every $i \in [m]$ and $j\in \set{1,2,3}$.            
        \end{enumerate}
    \end{enumerate}
Then, the set of objectives is $\beta = \set{\alpha_1,\alpha_2}\cup \set{\alpha_l : l\in X\cup \overline{X}}$, where $\alpha_1 = X\cup \overline{X}$, $\alpha_2 = \set{l_i^j : i\in [m], j\in \set{1,2,3}}$, and $\alpha_l = \set{l}\cup \set{l_i^j : i\in [m], j\in \set{1,2,3}, \text{ and }l_i^j= \overline{l}}$, for every $l\in X\cup \overline{X}$.

We prove the correctness of the construction. Note that when \pcov allocates both agents to the assignment sub-graph or both agents to the clause sub-graph, one of the objectives $\alpha_1$ and $\alpha_2$ is not satisfied in both plays, thus we only consider strategies in which \pcov allocates one agent (the assignment agent) to the assignment sub-graph, and one agent (the clause agent) to the clause sub-graph.

For the first direction, assume $\phi$ is satisfiable. Let $\xi: X\rightarrow \set{\bft,\bff}$ be a satisfying assignment, and let $F_1 = \tuple{f_1^1,f_1^2}$ be the strategy for \pcov such that $f_1^1$ is a strategy for the assignment agent that chooses from every variable vertex $v_i$ the literal vertex $x_i$ iff $\xi(x_i) = \bft$, and the literal vertex $\overline{x_i}$ iff $\xi(x_i) = \bff$; The strategy $f_1^2$ is a strategy for the clause agent that chooses from every clause vertex $C_i$ a clause-literal vertex $l_i^j$ such that $l_i^j$ is evaluated to $\bft$ in $\xi$. It is easy to see that $F_1$ is a covering strategy. Indeed, for every literal $l\in X\cup \overline{X}$, if $l$ is evaluated to $\bff$ in $\xi$, then the play in the assignment sub-graph does not visit the literal vertex $l$, thus $\alpha_l$ is satisfied in the play; and if $l$ is evaluated to $\bft$, then the play in the clause sub-graph does not visit clause-literal vertices that correspond to $\overline{l}$, thus $\alpha_l$ is satisfied in the play.

For the second direction, assume $\phi$ is not satisfiable. Consider a strategy $F_1 = \tuple{f_1^1,f_1^2}$ for \pcov, and assume WLOG that $f_1^1$ is for the assignment agent and $f_1^2$ is for the clause agent. Let $\xi:X\rightarrow\set{\bft,\bff}$ be an assignment to the variables in $X$ such that the assignment agent visits every literal vertex $l$ with $\xi(l) = \bft$ infinitely often. It is easy to see that there exists a literal $l\in X\cup \overline{X}$ with $\xi(l) = \bft$ such that the clause agent visits infinitely often clause-literal vertices that correspond to $\overline{l}$. Indeed, otherwise $\xi$ satisfies $\phi$. Thus, $\alpha_l$ is not satisfied in both plays, and so  $F_1$ is not a covering strategy.
\end{proof}
\noindent
When $|\beta|$ is fixed, the problem becomes easier also for $\gamma = C$. Intuitively, it follows from the fact that now, there is a fixed number of combinations of $k$ sets of objectives, one for each agent to satisfy on her own. Thus, finding a combination of them that covers $\beta$ can be done in NLOGSPACE. Formally, we have the following.  

\begin{thm}
\label{V2 = empty fixed beta}    
The coverage problem for \buchi or co-\buchi CGs with $V=V_1$ and fixed $|\beta|$ is NLOGSPACE-complete.
\end{thm}
 \begin{proof}
The results for \buchi follow from Theorem~\ref{V2 = emptyset and fixed k B}.
We continue to co-\buchi and start with the upper bound. Consider a co-\buchi CG $\G = \tuple{G,k,\beta}$ with $V = V_1$ and fixed $|\beta|$. An NLOGSPACE algorithm guesses a partition $\beta_1,\ldots,\beta_k$ of $\beta$ and writes it on the tape. Note that since $|\beta|$ is fixed, writing the partition can be done in logarithmic space. Then, for every $i\in [k]$, the algorithm checks that there exists a lasso path that satisfies the co-\buchi objective $\cup \beta_i$, which can be done in NLOGSPACE. 

Hardness in NLOGSPACE follows from an easy reduction from nonemptiness of co-\buchi word automata \cite{VW94}.
 \end{proof}
\noindent
We continue to the case where all the vertices in the game are owned by \pdis. 
There, the fact that \pcov has several agents plays no role, as all tokens are going to traverse the same play, chosen by \pdis. Consequently, searching for a disrupting strategy for \pdis has the flavor of checking the non-emptiness of nondeterministic automata with acceptance conditions that dualize generalized \buchi or co-\buchi conditions, and is easy.
 
\begin{thm}
\label{V1 = emptyset}
The disruption problem for \buchi or co-\buchi CGs with $V = V_2$ is NLOGSPACE-complete.
\end{thm}
\begin{proof}
Consider a CG $\G = \tuple{G, k, \beta}$ with $G=\zug{\emptyset,V,v_0,E}$.
Since \pdis directs all the tokens together, she has a disrupting strategy iff there is a path in $G$ that satisfies the Exists-$\tilde{\gamma}$ objective dual to $\beta$, namely a path that does not satisfy one of the objectives in $\beta$. 
We argue that the latter can be done in NLOGSPACE, for all $\gamma\in \set{\text{B,C}}$. 
Indeed, an NLOGSPACE algorithm can guess a set $\alpha \in \beta$ and then search for a lasso-shaped path that does not satisfy $\alpha$. 
In case $\gamma = \text{B}$, the vertices of the lasso loop should be disjoint from $\alpha$. In case $\gamma = \text{C}$, the algorithm guesses a vertex $v\in \alpha$ and checks that $v$ appears in the lasso loop. In both cases, it is possible to search for paths of length at most $|V|$.

This can be done in NLOGSPACE by guessing the vertices along the path one by one.

Hardness in NLOGSPACE follows from an easy reduction from the non-emptiness problem for nondeterministic ExistsC and ExistsB word automata \cite{VW94}. 
\end{proof}

\section{Properties of Coverage Games}
\label{prop cg}
In this section we study theoretical properties of CGs. We start with determinacy and show that CGs need not be determined. We then define and study the {\em decomposability} of objectives in CGs, namely the ability to a-priori partition the objectives among the agents. We argue that decomposability plays a key role in reasoning about CGs. 

\subsection{CGs need not be determined}
Recall that a CG is determined if \pcov has a covering strategy or \pdis has a disrupting strategy. We show that unlike two-player games, CGs need not be determined. Moreover, undeterminacy holds already for a CG with only two agents and three objectives. Note that this is tight, as CGs with a single agent or only two objectives belong to the $k=1$ or $k\geq |\beta|$ cases studied in Section~\ref{k easy}, and are thus determined. 

Intuitively, the undeterminacy of CGs follows from the incomplete information of \pcov's agents about all the plays in the outcome, as well as the inability of \pdis to make use of this incomplete information.

Formally, we have the following. 

\begin{thm}
\label{k < m win undetermined} 
For all $\gamma \in \set{\text{B,C}}$, we have that $\gamma$-CGs need not be determined. Moreover, undeterminacy holds already for a $\gamma$-CG with only two agents and three objectives.
\end{thm}
\begin{proof}
For all $\gamma \in \set{\text{B,C}}$, we describe a $\gamma$-CG $\G = \tuple{G,2, \set{\alpha_1,\alpha_2,\alpha_3}}$ such that neither \pcov has a covering strategy, nor \pdis has a disrupting strategy. The game graph $G$ appears in Figure~\ref{fig undetermined}. 
\begin{figure}[htp]
\centering
\includegraphics[width=7cm]{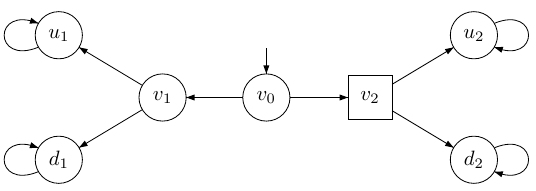}
\caption{An undetermined coverage game.}
\label{fig undetermined}
\end{figure}

The objectives $\alpha_1,\alpha_2$, and $\alpha_3$ depend on $\gamma$. 
We start with $\gamma = B$. 
There, we define $\alpha_1 = \set{u_1, u_2}$, $\alpha_2 = \set{d_1, d_2}$, and $\alpha_3 = \set{u_2, d_{2}}$.
Intuitively \pcov has no covering strategy as $\alpha_3$ requires one agent to move to $v_2$, and no matter how \pcov directs the other agent, \pdis can move the tokens that reach $v_2$ so that only one of $\alpha_1$ and $\alpha_2$ is satisfied. Also, \pdis does not have a disrupting strategy, as no matter how the strategy moves token that reach $v_2$, \pcov can direct the agents so that all objectives are covered. 

We first prove that every strategy $F_1 = \tuple{f_1^1,f_1^2}$ for \pcov is not a covering strategy in $\G$. If $f_1^1(v_0) = f_1^2(v_0) = v_1$, then $F_1$ satisfies at most $\alpha_1$ and $\alpha_2$; If $f_1^1(v_0) = f_1^2(v_0) = v_2$, then $F_1$ guarantees the satisfaction of the two objectives $\alpha_1$ and $\alpha_3$, when \pdis proceeds from $v_2$ to $u_2$, and the two objectives $\alpha_2$ and $\alpha_3$, when \pdis proceeds from $v_2$ to $d_2$; Otherwise, without loss of generality, we have that $f_1^1(v_0) = v_1$ and $f_1^2(v_0) =v_2$. If $f_1^1(v_1) = u_{1}$, then when \pdis proceeds from $v_2$ to $u_{2}$, only $\alpha_1$ and $\alpha_3$ are satisfied, and in a similar way, if $f_1^1(v_1) = d_{1}$, then when \pdis proceeds from $v_2$ to $d_{2}$, only $\alpha_2$ and $\alpha_3$ are satisfied. Therefore, $F_1$ is not a covering strategy.

We now prove that every strategy $f_2$ for \pdis is not a disrupting strategy. If $f_2(v_2) = u_{2}$, then when \pcov uses the strategy in which one agent goes to $d_{1}$ and one agent goes to $v_2$, all three objectives are satisfied; in a similar way, if $f_2(v_2) = d_{2}$, then when \pcov uses the strategy in which one agent goes to $u_{1}$ and one agent goes to $v_2$, all three objectives are satisfied. Therefore, $f_2$ is not a disrupting strategy.

When $\gamma = C$, we define $\alpha_1 = \set{u_1, u_2}$, $\alpha_2 = \set{d_1, d_{2}}$, and $\alpha_3 = \set{u_{1}, d_{1}}$. 
Since the only edges from the vertices $u_1,d_1,u_2$, and $d_2$ are self-loops, the obtained co-\buchi game coincides with the \buchi game analyzed above, and we are done.
\end{proof}

\subsection{Decomposability of objectives}
A key challenge for \pcov in winning a CG $\G = \tuple{G,k, \beta}$ with $2 \leq k < |\beta|$ is the need to partition the objectives in $\beta$ among her $k$ agents. In this section we show that in general, \pcov cannot {\em a-priori} partition $\beta$ among the $k$ agents. Moreover, decompositions of $\beta$ are related to decompositions of $G$, which are the key to our algorithms for reasoning about CGs. 

To see the idea behind the notion of decomposition, recall that when $k \geq |\beta|$, \pcov can assign to Agent~$i$ the task of satisfying the $i$-th objective in $\beta$. Such an assignment can be done before the plays start. 
On the other hand, consider the \buchi CG $\G$ from Example~\ref{example}, and 
recall that \pcov has a covering strategy in $\G$. In this covering strategy, the two agents of \pcov cover three objectives. Since \pcov does not know in advance whether the plays would stay forever in the left sub-game, or would eventually reach the vertices $v_1$ or $v_2$, she cannot a-priori assign the two agents with the objectives they aim to cover. Indeed (see details in Example~\ref{example2}), only after the game reaches $v_1$ or $v_2$, \pcov can decide whether the agent that covers $\alpha_2$ (green) is the same agent that covers $\alpha_1$ (red) -- which is the case when the plays reach $v_1$, or is the same agent that covers $\alpha_3$ (yellow) -- which is the case when the plays reach $v_2$. 

We formalize the ability of \pcov to make an a-priori partition of the objectives by means of {\em decompositions}, defined below. 
Consider a game graph $G=\zug{V_1,V_2,v_0,E}$. For a vertex $v \in V_1$, let $G^v=\zug{V_1,V_2,v,E}$ be $G$ with initial vertex $v$. Consider a CG $\G = \tuple{G,k,\beta}$, with $\beta = \set{\alpha_1,\ldots,\alpha_m}$. Let $A=[k]$. For $1\leq l \leq k$ and a vertex $v \in V_1$, we say that $\beta$ is {\em $(k,l)$-decomposable} in $v$ if \pcov can decompose the task of covering $\beta$ in $G^v$ by the agents in $A$ to $l$ sub-tasks, each assigned to a different subset of $A$. Formally, $\beta$ is $(k,l)$-decomposable in $v$ if there exists a partition $\beta_1,\ldots,\beta_l$ of $\beta$ and a partition $A_1,\ldots,A_l$ of $A$ to nonempty sets, such that for every $i\in [l]$, \pcov wins the CG $\G^v_i = \tuple{G^v,|A_i|,\beta_i}$, namely the game that starts in $v$ and in which the agents in $A_i$ have to cover the objectives in $\beta_i$. When $v=v_0$, we say that $\beta$ is {\em $(k,l)$-decomposable} in $G$.

We start with some easy observations. First, note that, by definition, $\beta$ is {\em $(k,1)$-decomposable} in $G$ iff \pcov wins $\G$. Indeed, when $l=1$, we have that $A_1=A$ and $\beta_1=\beta$, thus we do not commit on a decomposition of $\beta$, and the definitions coincide. Then, $\beta$ is {\em $(k,k)$-decomposable} in $G$ iff
there exists a partition of $\beta$ to $k$ sets $\beta_1,\ldots,\beta_k$ such that for every $1 \leq i\leq k$, \pcov has a strategy that satisfies the All objective $\beta_i$. Indeed, when $l=k$, each of the sets $A_i$ is a singleton, thus the single agent in $A_i$ has to satisfy all the objectives in $\beta_i$. Finally, the bigger $l$ is, the more refined is the partition of $\beta$ to which \pcov commits, making her task harder. Formally, we have the following (see Example~\ref{example} for a CG with no a-priori decomposition): 

\begin{thm}
\label{a priori is hard}
For all $\gamma \in \set{\text{B,C}}$, we have the following.
\begin{itemize}
\item
Consider a $\gamma$-CG $\G = \tuple{G,k,\beta}$. For every $l\in [k]$ and vertex $v$ of $G$, if $\beta$ is $(k,l)$-decomposable in $v$, then $\beta$ is $(k,l')$-decomposable in $v$, for all $l' \leq l$. 
\item For every $k\geq 2$, there exists a $\gamma$-CG $\G = \tuple{G,k,\beta}$ with $|\beta| = k+1$ such that \pcov wins $\G$ but $\beta$ is not $(k,l)$-decomposable in $G$, for all $l>1$.
\item For every $k\geq 2$ and $1\leq l < k$, there exists a $\gamma$-CG $\G = \tuple{G,k,\beta}$ such that $\beta$ is $(k,l)$-decomposable in $G$ but is not $(k,l+1)$-decomposable in $G$.
\end{itemize}
\end{thm}
\begin{proof}
\label{app a priori is hard}
For the first claim, consider a CG $\G = \tuple{G,k,\beta}$, and consider $1 <  l \leq k$ and a vertex $v$ of $G$ such that $\beta$ is $(k,l)$-decomposable in $v$. We show that $\beta$ is also $(k,l-1)$-decomposable in $v$. The claim for all $l' \leq l$ then follows by repeated (possibly zero) applications of this claim. 
Since $\beta$ is $(k,l)$-decomposable in $v$, there exists a partition $\beta_1,\ldots,\beta_l$ of $\beta$ and a partition $A_1,\ldots,A_l$ of $A$ to nonempty sets, such that for every $1\leq i\leq l$, \pcov wins the game $\tuple{G^v,|A_i|,\beta_i}$. It follows that \pcov also wins the CG $\tuple{G^v, |A_{l-1}\cup A_l|, \beta_{l-1}\cup \beta_l}$. Indeed, the strategy in which the agents in $A_{l-1}$ use their winning strategy for $\beta_{l-1}$ and the agents in $A_l$ use their winning strategy for $\beta_l$ covers $\beta_{l-1}\cup \beta_l$. Hence, $\beta_1,\ldots,\beta_{l-2},\beta_{l-1}\cup \beta_l$ and $A_1,\ldots, A_{l-2}, A_{l-1}\cup A_l$ are partitions of $\beta$ and $A$, respectively, to $l-1$ nonempty sets, witnessing that $\beta$ is $(k,l-1)$-decomposable from $v$.

For the second item, consider $k\geq 2$. We prove the claim for $l=2$. By the first item, the result then follows for all $l > 1$.
First, for $\gamma=B$, consider the $\gamma$-CG $\G = \tuple{G,k,\beta}$, where $G$ is as follows (see Figure~\ref{fig k 2 decomp} for an example for $k=2$). Starting from the initial vertex $v_0$, \pdis chooses between $k+1$ vertices $v_1,\ldots,v_{k+1}$. From $v_i$, \pcov chooses among $k$ self-looped sinks $s_i^1,\ldots,s_i^k$, for every $i\in [k+1]$. Then, for every $i\in [k+1]$ and $j \in [i-1]$, the sink $s_i^j$ satisfies the objective $\alpha_j$, and for every $j \in [k+1]\setminus [i]$, the sink $s_i^{j-1}$ satisfies the objective $\alpha_j$. The sink $s_i^1$ also satisfies $\alpha_i$. That is, every objective $\alpha_j \in \beta \setminus \set{\alpha_i}$ is satisfied in a separate sink reachable from $v_i$. Formally, $\beta = \set{\alpha_i: i\in [k+1]}$, with $\alpha_i = \set{s_j^i : i+1\leq j\leq  k+1}\cup \set{s_j^{i-1} : 1\leq j\leq i-1} \cup \set{s_i^1}$, for every $i\in [k+1]$. 

\begin{figure}[htp]
\centering
\includegraphics[width=8cm]{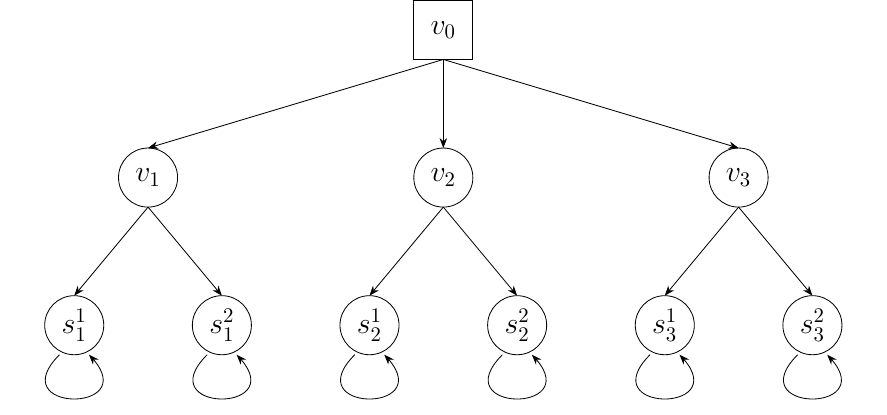}
\caption{The game graph $G$ for $k=2$. Here, $\beta = \set{\alpha_1,\alpha_2,\alpha_3}$, with $\alpha_1 = \set{s_2^1,s_3^1,s_1^1}$, $\alpha_2 = \set{s_3^2,s_1^1,s_2^1}$, and $\alpha_3 = \set{s_1^2,s_2^2,s_3^1}$.
}
\label{fig k 2 decomp}
\end{figure}

\noindent
It is not hard to see that \pcov wins $\G$. Indeed, $\beta$ can be satisfied in $G^{v_{i}}$, for every successor $v_i$ of $v_0$.
We show that $\beta$ is not $(k,2)$-decomposable in $G$. Assume by contradiction otherwise. Thus, there exist $\beta'\subset \beta$ and $A'\subset A$ such that \pcov wins $\tuple{G,|A'|,\beta'}$ and $\tuple{G,|A\setminus A'|, \beta\setminus \beta'}$. Note that every winning strategy for \pcov in $\tuple{G,|A'|,\beta'}$ requires $|\beta'|$ agents. Indeed, since $|\beta'| \leq  k$, \pdis can choose from the initial vertex a successor $v_i$ such that $\alpha_i\notin \beta'$, thus $\beta'$ is covered from $v_i$ iff $|\beta'|$ agents are sent to $|\beta'|$ different successors of $v_i$. In a similar way, every winning strategy for \pcov in $\tuple{G,|A\setminus A'|,\beta\setminus\beta'}$ requires $|\beta\setminus \beta'|$ agents. Therefore, $k = |A| = |A'|+|A\setminus A'| \geq |\beta'| + |\beta\setminus \beta'| = |\beta| = k+1$, contradiction.

Next, for $\gamma = C$, we define $\G' = \tuple{G, k, \beta'}$ over the same game graph, now with $\beta' = \set{\alpha_1',\ldots,\alpha_{k+1}'}$, where $\alpha_i' = \cup(\beta \setminus \set{\alpha_i})$, for every $i\in [k+1]$. Since the vertices $s_i^j$ are sinks, every play in $G$ satisfies the $\gamma$ objective $\alpha_i'$ iff it satisfies the $\Tilde{\gamma}$ objective $\alpha_i$, for every $i\in [k+1]$, we have that \pcov wins $\G'$ but $\beta'$ is not $(k,2)$-decomposable in $G$.

We continue to the last claim. 
Consider $k\geq 2$ and $1\leq l < k$. First, for $\gamma = B$, we describe a $\gamma$-CG $\G = \tuple{G,k,\beta}$ with $\beta = \set{\alpha_1,\ldots,\alpha_{k+1}}$ such that $\beta$ is $(k,l)$-decomposable in $G$, but $\beta$ is not $(k,l+1)$-decomposable in $G$. The game is as follows. Let $\G' = \tuple{G',k-(l-1),\beta'}$ be a CG with $\beta' = \set{\alpha'_1,\ldots,\alpha'_{k-(l-1)+1}}$ such that \pcov wins $\G'$, but $\beta'$ is not $(k-(l-1),2)$-decomposable in $G'$. By item 2, such a CG exists. Starting from the initial vertex in $G$, \pdis chooses between $l$ successors $v_1,\ldots,v_l$. The vertices $v_1,\ldots,v_{l-1}$ are self-looped sinks, and from $v_l$, \pcov proceeds to the sub-graph $G'$. 
Then, we define $\alpha_i = \set{v_i}$, for every $i\in [l-1]$, and $\alpha_i = \alpha'_{i - (l-1)}$, for every $l\leq i\leq k+1$. 
It is not hard to see that $\beta$ is $(k,l)$-decomposable in $G$. Indeed, for every $i\in [l-1]$, \pcov wins $\tuple{G,1,\set{\alpha_i}}$, and also \pcov wins $\tuple{G, k-(l-1),\beta'}$, since \pcov wins $\G'$. 
We continue to show that $\beta$ is not $(k,l+1)$-decomposable. Assume by contradiction otherwise, and let $\beta_1,\ldots,\beta_{l+1}$ and $A_1,\ldots,A_{l+1}$ be partitions of $\beta$ and $A$ to $l+1$ nonempty sets, respectively, such that \pcov wins $\tuple{G,|A_i|,\beta_i}$, for every $1\leq i\leq l+1$. Since there are only $l-1$ objectives in $\beta \setminus \beta'$, we have that $\beta'$ is divided among at least two different sets in the partition, implying an a-priori partition of $\beta'$, contradicting the fact $\beta'$ is not $(k-(l-1),2)$-decomposable in $G'$. The same results hold for $\gamma = C$ and the $\gamma$-CG $\G' = \tuple{G, k, \beta''}$ with $\beta'' = \set{\alpha_1'',\ldots,\alpha_{k+1}''}$ such that $\alpha_i'' = \cup(\beta \setminus \set{\alpha_i})$, for every $i\in [k+1]$. 
\end{proof}
\noindent
Consider a vertex $v\in V_1$.
For $l\in[k]$, we say that $v$ is a {\em $(k,l)$-fork} for $\beta$ if \pcov has a covering strategy $F_1 = \tuple{f_1^1,\ldots,f_1^k}$ in $G^v$ such that different agents are sent to $l$ different successors of $v$. Thus, there exist $l$ successors $v_1,\ldots,v_l$ of $v$, and for all $i\in[l]$, there is at least one agent $j \in [k]$ such that $f_1^j(v)=v_i$. Note that we do not require $f_1^i$ to be memoryless and refer here to the strategy in the first round in $G^v$. The following lemma then follows from the definitions. 

\begin{lem}
\label{fork implies decomposable}
Consider a CG $\G = \tuple{G,k,\beta}$. For every vertex $v \in V_1$ of $G$ and $l\in[k]$, if $v$ is a $(k,l)$-fork for $\beta$, then $\beta$ is $(k,l)$-decomposable in $v$.
\end{lem}
\begin{proof}
Consider a vertex $v \in V_1$ and $1 \leq l \leq k$ such that $v$ is a $(k,l)$-fork for $\beta$. Let $F_1 = \tuple{f_1^1,\ldots,f_1^k}$ be a winning strategy for \pcov that sends the agents to $l$ different successors $v_1,\ldots,v_l$ of $v$. For every $1\leq i\leq l$, let $A_i$ be the set of agents $j\in[k]$ such that $f_1^j(v) = v_i$, and let $F_1^i = \set{f_1^j}_{j\in A_i}$ be the set of strategies of the agents sent to $v_i$. 
Let $\beta_i$ be the set objectives of $\beta$ that $F_1^i$ covers. That is, $\alpha_j \in \beta_i$ iff for all strategies $f_2$ of \pdis, there is a play in $\outcome(F_1^i,f_2)$ that satisfies $\alpha_j$. Equivalently, $\beta_i \subseteq \beta$ is the maximal subset of $\beta$ such that $F_1^i$ is a winning strategy for \pcov in $\tuple{G_{v_i}, |A_i|, \beta_i}$. It is then sufficient to show that $\bigcup_{i\in[l]}\beta_i = \beta$. Assume by contradiction that there exists an objective $\alpha_j \in \beta$ such that $\alpha_j\notin \beta_i$, for every $i\in [l]$, which implies there exists a strategy $f_2^i$ for \pdis in $G_{v_i}$ such that $\alpha_j$ is not covered in $\outcome(\tuple{F_1^i,f_2^i})$. Let $f_2$ be the strategy for \pdis from $v$ that follows $f_2^i$ from $v_i$, for every $i\in[l]$. Since $\alpha_j$ is not covered in $\outcome(\tuple{F_1,f_2})$, we have that $F_1$ is not a winning strategy, and we have reached a contradiction. 
\end{proof}

\noindent
We say that a vertex $v \in V_1$ is a {\em fork\/} if it is a $(k,l)$-fork for $\beta$, for some $2\leq l\leq k$. We denote by $F \subseteq V_1$ the set of vertices that are forks. 
Let $V_{avoid}\subseteq V$ be the set of vertices from which \pdis can avoid reaching forks. That is, $v\in V_{avoid}$ if there exists a strategy $f_2$ for \ptwo from $v$ such that for every strategy $f_1$ for \pone, the play $\outcome(\tuple{f_1,f_2})$ does not reach $F$. Finally, let $G_{avoid} = \tuple{V_1',V_2',v_0,E'}$ be the sub-graph of $G$ with vertices in $V_{avoid}$. That is, $V_1' = V_1 \cap V_{avoid}$, $V_2' = V_2 \cap V_{avoid}$, and $E' = E\cap (V_{avoid}\times V_{avoid})$.

\begin{exa}
\label{example2}
Consider the \buchi CG $\G$ from Example~\ref{example}. 
Recall that \pcov has a covering strategy in $\G$. We claim that the only forks for $\beta$ in $G$ are $v_1$ and $v_2$. Thus, a winning strategy for \pcov cannot a-priori decompose the objectives in $\beta$. Note that this is the case also for the covering strategy described in Example~\ref{example}. 

Indeed, in $v_1$, \pcov can assign both $\alpha_1$ and $\alpha_2$ to an agent that proceeds to $u_2$ and assign $\alpha_3$ to an agent that proceeds to $d_2$. Similarly, in $v_2$, \pcov can assign $\alpha_1$ to an agent that proceeds to $u_3$ and assign both $\alpha_2$ and $\alpha_3$ to an agent that proceeds to $d_3$. 
Also, as \pdis may direct tokens to either $v_1$ or $v_2$, and the above decompositions are different and are the only possible decompositions in $v_1$ and $v_2$, \pcov cannot a-priori decompose $\beta$. Hence, $v_0$ is not a fork. 
Finally, as $\beta$ is not covered in each of $u_2$, $d_2$, $u_3$, and $d_3$, they are not forks either. 

Note also that since \pdis can direct tokens from $u_1,m_1$, and $d_1$ back to $v_0$, the graph $G_{avoid}$ is the sub-graph of $G$ whose vertices are $v_0,u_1,m_1$ and $d_1$. 
\end{exa}

\begin{thm}
\label{covering str characterization}
Consider a CG $\G = \tuple{G,k,\beta}$, and a vertex $v\in V$. Then, \pcov has a covering strategy in $\tuple{G^v,k,\beta}$ iff $v\notin V_{avoid}$, or $v\in V_{avoid}$ and \pcov wins $\tuple{G_{avoid}^v,1,\beta}$.
\end{thm}

\begin{proof}
We first prove that for all vertices $v \notin V_{avoid}$, \pcov has a strategy to cover $\beta$ from $v$. Consider a vertex
 $v \notin V_{avoid}$. By the definition of $V_{avoid}$, \pcov has a strategy to reach a fork from $v$. Hence, \pcov can cover $\beta$ from $v$ by letting all agents follow the same strategy until some $(k,l)$-fork $u$ for $\beta$ is reached. By Lemma~\ref{fork implies decomposable}, $\beta$ is $(k,l)$-decomposable in $u$. Hence, once the $k$ tokens of the agents reach $u$, \pcov can cover $\beta$ by decomposing it. Let $F_1^v$ denote a covering strategy for \pcov from a vertex $v \notin V_{avoid}$. 

We argue that for every vertex $v\in V_{avoid}$, \pcov has a strategy to cover $\beta$ from $v$ iff \pcov wins $\tuple{G_{avoid}^v,1,\beta}$.
Consider a vertex $v\in V_{avoid}$, and assume first that \pcov has a covering strategy from $v$. Since \pdis can force the play to stay in $G_{avoid}$, every covering strategy for \pcov from $v$ in $G$ is also a covering strategy from $v$ in $G_{avoid}$. Also, for every covering strategy $F_1 = \tuple{f_1^1,\ldots,f_1^k}$ from $v$, the strategies for the different agents agree with each other as long as the play stays in $G_{avoid}$. That is, $f_1^1(h) = \cdots = f_1^k(h)$, for every $h\in {V^*_{avoid}}\cdot V_1'$. Indeed, vertices from which there exists a covering strategy for \pcov that sends different agents to different successors are forks, and by the definition of $V_{avoid}$, the sub-graph $G_{avoid}$ does not contain forks. 
Hence, the strategy $f_1: {V^*_{avoid}}\cdot V_1'\rightarrow V_{avoid}$ with $f_1(h) = f_1^1(h)$ for every $h\in {V^*_{avoid}}\cdot V_1'$ is a covering strategy for \pcov $\tuple{G_{avoid}^v,1,\beta}$.

For the second direction, assume that \pcov wins $\tuple{G_{avoid}^v,1,\beta}$. We show that \pcov has a covering strategy from $v$. Let $f_1$ be a covering strategy for \pcov in $\tuple{G_{avoid}^v,1,\beta}$. Consider the strategy $F_1 = \tuple{f_1^1,\ldots,f_1^k}$ in which, for all $1 \leq i \leq k$, the strategy $f_1^i$ agrees with $f_1$ as long as the play stays in $G_{avoid}$, and proceeds with $F_1^{u}$ once the play (that is, all tokens together, as this may happen only in a transition taken by \pdis), leaves $G_{avoid}$ and reaches a vertex $u \not \in V_{avoid}$. For every strategy $f_2$ for \pdis, either all the outcomes of $F_1$ and $f_2$ stay in $G_{avoid}$, in which case $\beta$ is covered by each of the agents, or all the outcomes of $F_1$ and $f_2$ leave $G_{avoid}$, in which case they follow strategies that cover $\beta$ by decomposing it. 
\end{proof}

\section{The Complexity of the Coverage Problem}
\label{sec comp cov}
In this section, we study the complexity of the coverage problem for \buchi and co-\buchi CGs, and show that it is PSPACE-complete. 

We study the joint complexity, namely in terms of $G, \beta$, and $k$. In Sections~\ref{CGs with a fixed number of agents} and~\ref{CGs with a fixed number of objectives}, we complete the picture with a parameterized-complexity analysis. 

We start with upper bounds. 
Consider a CG $\G = \tuple{G,k,\beta}$. 
Recall the set $F$ of forks for $\beta$, and the game graph $G_{avoid}$, in which \pdis can avoid reaching forks. By Theorem~\ref{covering str characterization}, \pcov has a covering strategy in $G$ iff \pone has a winning strategy in $G_{avoid}$ for the All objective $\beta$, and when the play leaves $G_{avoid}$ (or if the initial vertex is not in $G_{avoid}$), \pcov can cover $\beta$ by reaching a fork, from which $\beta$ can be decomposed.

The above characterization is the key to our algorithm for deciding whether \pcov has a covering strategy. Essentially, the algorithm guesses the set $U$ of forks, checks that \pone wins the game with the All objective $\beta$ in the sub-graph induced by $U$, and checks (recursively) that the vertices in $U$ are indeed forks. 

The checks and the maintenance of the recursion require polynomial space. 

Since NPSPACE=PSPACE, we have the following.\footnote{It is not hard to see that the proof of Theorem~\ref{decide cover buchi upper} applies to all prefix-independent objective types $\gamma$ such that All-$\gamma$ games can be decided in PSPACE.}
\begin{thm}
\label{decide cover buchi upper}
The coverage problem for \buchi or co-\buchi CGs can be solved in PSPACE.
\end{thm}
\begin{proof}
We describe a non-deterministic Turing machine (NTM) $T$ that runs in polynomial space, such that $T$ accepts a CG $\G$ iff \pcov has a covering strategy in $\G$.

Given a CG $\G = \tuple{G,k,\beta}$, the NTM $T$ guesses a set of vertices $U\subseteq V$, and checks that they are forks. To check that a vertex $u\in U$ is a fork, the NTM guesses a partition $\beta_1,\ldots,\beta_l$ of $\beta$ and a partition $A_1,\ldots,A_l$ of $[k]$ to nonempty sets, for some $2\leq l\leq k$, and checks recursively whether \pcov has a covering strategy in $\tuple{G^u,|A_i|,\beta_i}$, for every $i\in[l]$. If one of the recursive checks fails, the NTM rejects. Otherwise, $T$ calculates the restriction $G'$ of $G$ to vertices from which \pdis can force the play to avoid $U$. Then, $T$ accepts if the initial vertex $v$ of $G$ is not in $G'$ or \pone wins from $v$ in the All-$\gamma$ game $\tuple{G',\beta}$. Otherwise, $T$ rejects.

We prove the correctness of the construction. That is, we prove that $T$ accepts a game $\G$ iff \pcov has a covering strategy in $\G$. 
By Theorem~\ref{covering str characterization}, \pcov has a covering strategy from a vertex $v$ iff $v\notin V_{avoid}$, or $v\in V_{avoid}$ and \pcov wins $\tuple{G_{avoid}^v,1,\beta}$. So, if \pcov has a covering strategy, the NTM can guess the set $F$ of forks for the current set $\beta$ of objectives and number $k$ of agents. Then, $G'$ coincides with $G_{avoid}$, and by Theorem~\ref{covering str characterization}, the NTM accepts. For the second direction, assume the NTM accepts. Then, the vertices in $U$ are forks by definition, and \pcov wins $\tuple{G',1,\beta}$. By Theorem~\ref{covering str characterization}, it is enough to show that if $v\in G_{avoid}$, then \pcov wins $\tuple{G_{avoid}^v,1,\beta}$. Since the vertices in $U$ are forks, $G_{avoid}$ is a sub-graph of $G'$. So, since \pdis can force the play to stay in $G_{avoid}$ if $v$ is in it, we also have that \pcov wins $\tuple{G_{avoid}^v,1,\beta}$.

To complete the proof, we show that the NTM runs in polynomial space. Note that only a polynomial number of recursive checks can be made before reaching an empty set of objectives, and each check requires a polynomial space for guessing $U$ and the appropriate partitions. 
Also, calculating $G'$ can be done in polynomial time, and checking whether \pone wins an All-$\gamma$ game can be done in polynomial time as well, for both $\gamma\in\set{\text{B,C}}$ \cite{VW86a}. 
Thus, the NTM runs in polynomial space. 
\end{proof}

\noindent
Note that the algorithm can return a covering strategy (if one exists). Such a strategy can be described by the set $F$ of forks, the sub-game $G_{avoid}$, a covering strategy in $\tuple{G_{avoid},1,\beta}$, and the strategies for reaching forks. Then, for each fork, the strategy describes the partitions of $\beta$ and $[k]$, and the (recursive) description of the strategies in the decomposed game. Note that the description may require more than polynomial space.

We continue to the lower bounds.
\begin{thm}
\label{decide cover lower}
The coverage problem for \buchi or co-\buchi CGs is PSPACE-hard.
\end{thm}
\begin{proof}
We describe reductions from QBF. We start with \buchi CGs. Consider a set of variables $X = \set{x_1,\ldots,x_n}$ and a QBF formula $\Phi = Q_1 x_1,\ldots,Q_n x_n \phi$ where $\phi = C_1\wedge \cdots \wedge C_m$, with $C_i = (l_i^1\vee l_i^2\vee l_i^3)$, for every $i\in [m]$. We construct a \buchi CG $\G = \tuple{G,|X|,\beta}$ such that $\Phi = \bft$ iff \pcov has a covering strategy in $\G$.

The game graph $G$ (see example in Figure~\ref{fig QBF to cover buchi}) lets the players choose an assignment to the variables in $X$, following the quantification order of $\Phi$ and starting from the outermost variable. 
Starting from the initial vertex, when the play reaches a vertex that corresponds to an existential variable $x_i$, \pcov chooses between choosing an assignment to $x_i$, which involves proceeding to a self-looped sink that corresponds to the literal $x_i$, or proceeding to a self-looped sink that corresponds to the literal $\overline{x_i}$, and proceeding to the next variables. When the play reaches a vertex that corresponds to a universal variable $x_i$, \pdis chooses an assignment to $x_i$, by proceeding to a vertex that corresponds to the literal $x_i$, or proceeding to a vertex that corresponds to the literal $\overline{x_i}$. In those vertices, \pcov chooses between staying in the current vertex, and proceeding to the next variables.

\begin{figure}[htp]
\centering
\includegraphics[width=7.6cm]{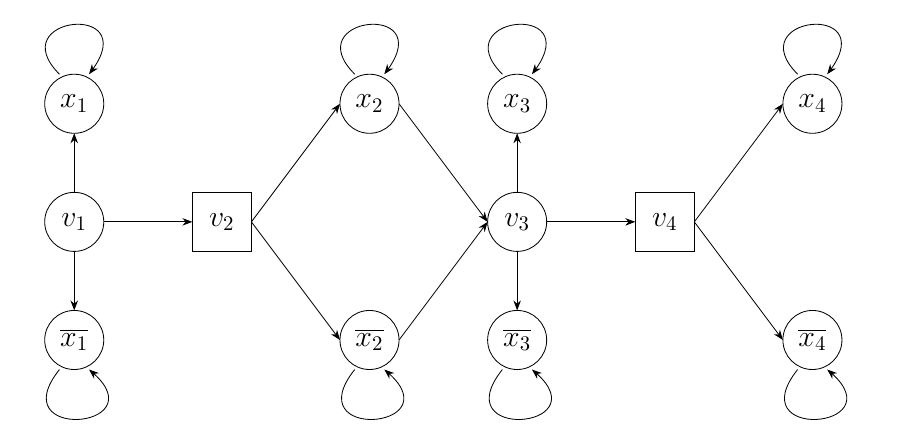}
\caption{The game graph $G$ for $\Phi = \exists x_1 \forall x_2 \exists x_3 \forall x_4 \phi$.
}
\label{fig QBF to cover buchi}
\end{figure}

\noindent
Since the vertices that correspond to existential variables are self-looped sinks, the choices of \pcov are not reflected in the history of plays that reach the subsequent variables. Hence, when \pdis chooses an assignment to a universal variable, she is unaware of the assignments \pcov chose for preceding existential variables. The vertices that correspond to universal variables are not sinks, thus when choosing assignments to existential variables, \pcov is aware of the assignments \pdis chose for preceding universal variables. 

Then, for every variable $x\in X$, we define an objective $\alpha_x$ that forces \pcov to allocate an agent to stay in the vertices that correspond to the assignment chosen for $x$. Also, for every clause $C_i$, we define the objective $\alpha_{C_{i}}$ that consists of vertices that correspond to $C_i$'s literals. Thus, the joint assignment satisfies $\phi$ iff \pcov can cover all the objectives. 

Formally, $\G = \tuple{G,|X|,\beta}$, where the game graph $G = \tuple{V_1,V_2,v_1,E}$ has the following components.
\begin{enumerate}
\item $V_1 = \set{v_i: i\in [n] \text{ and }Q_i = \exists}\cup \set{l : l\in X\cup \overline{X}}$. The vertices in $\set{v_i: i\in [n] \text{ and }Q_i = \exists}$ are {\em existential variable vertices}, and the vertices in $\set{l : l\in X\cup \overline{X}}$ are {\em literal vertices}.
\item $V_2 = \set{v_i: i\in [n] \text{ and }Q_i = \forall}$. The vertices in $\set{v_i: i\in [n] \text{ and }Q_i = \forall}$ are {\em universal variable vertices}.
\item The set $E$ of edges includes the following edges.
\begin{enumerate}
\item $\tuple{v_i, x_i}$ and $\tuple{v_i,\overline{x_i}}$, for every $i\in [n]$. That is, from the variable vertex $v_i$, the owner of the vertex chooses an assignment to the variable $v_i$ by proceeding to one of the literal vertices $x_i$ and $\overline{x_i}$, which correspond to assigning $\bft$ and $\bff$, respectively. 
\item $\tuple{l,l}$, for every $l\in X\cup \overline{X}$. That is, \pcov can choose to stay in every literal vertex.
\item $\tuple{v_i,v_{i+1}}$, for every $i\in [n-1]$ such that $Q_i = \exists$. That is, from existential variable vertices, \pcov can proceed to the next variable vertex.
\item $\tuple{l, v_{i+1}}$, for every $i\in [n-1]$ such that $Q_i = \forall$ and $l\in \set{x_i,\overline{x_i}}$. That is, from literal vertices that correspond to universal variables, \pcov can choose between staying in the literal vertex, and proceeding to the next variable vertex.
\end{enumerate}
\end{enumerate}
The set $\beta$ contains the following objectives.
\begin{enumerate}
    \item $\alpha_x = \set{x,\overline{x}}$, for every $x\in X$.
    \item $\alpha_{C_i} = \set{l_i^1,l_i^2,l_i^3}$, for every $i\in [m]$.
\end{enumerate}

\noindent
We prove the correctness of the construction. That is, we prove that $\Phi = \bft$ iff \pcov has a covering strategy in $\G$. Note that the objective $\alpha_x$, for every $x\in X$, guarantees that exactly one agent has to stay in a vertex that correspond to one of the literals $x$ and $\overline{x}$.

For the first direction, assume $\Phi =\bft$. Consider the strategy $F_1$ for \pcov that chooses an assignment to the existential variables, depending on the assignments chosen for preceding variables, in a way that ensures $\phi$ is satisfied. Then, for every literal vertex that corresponds to universal variables the play encounters, $F_1$ leaves one of the agents to traverse its self-loop of indefinitely. Thus, the objective $\alpha_x$ is satisfied in the play that stays in the literal vertex $x$ or $\overline{x}$, for every $x\in X$. And, since the chosen assignment satisfies $\phi$, and literal vertices are visited infinitely often iff they are evaluated to $\bft$ in the chosen assignment, every objective $\alpha_{C_i}$ is satisfied in at least one play. Therefore, $F_1$ is a covering strategy.  

For the second direction, assume there exists a covering strategy $F_1$ for \pcov in $\G$. Note that by the definition of the construction, $F_1$ allocates a single agent to stay indefinitely in every literal vertex the play reaches. Also note that for every strategy $f_2$ for \pdis, the set of literal vertices that the plays visit infinitely often induces an assignment that satisfies $\phi$. Indeed, otherwise there exists a strategy $f_2$ for \pdis that induces an assignment for which there exists a clause $C_i$ all whose literals are evaluated to $\bff$, thus $\alpha_{C_i}$ is not covered in $\outcome(\tuple{F_1,f_2})$, contradicting the fact that $F_1$ is a covering strategy. Thus, there exists an assignment to the existential variables, depending on assignments to preceding universal variables, that guarantees the satisfaction of $\phi$. Accordingly, $\Phi = \bft$.

For co-\buchi CGs, note that, as in previous reductions, a \buchi objective $\alpha\subseteq (X\cup \overline{X})$ is satisfied in a play $\rho$ iff the co-\buchi objective $(X\cup \overline{X})\setminus \alpha$ is satisfied in $\rho$, and thus the same reduction with dual objectives applies. 
Hence, \pcov has a covering strategy in the \buchi CG $\G=\tuple{G,|X|,\beta}$ iff \pcov has a covering strategy in the co-\buchi CG $\G' = \tuple{G,|X|,\beta'}$, with $\beta' = \set{(X\cup \overline{X})\setminus \alpha: \alpha\in \beta}$. Therefore, $\Phi = \bft$ iff \pcov has a covering strategy in $\G'$.
\end{proof}

\section{The Complexity of the Disruption Problem}
In this section, we study the complexity of the disruption problem for \buchi and co-\buchi CGs. Note that since CGs are undetermined, the coverage and disruption problems are not dual, and so the results in Section~\ref{sec comp cov} do not imply that the disruption problem is PSPACE-complete. In fact, we show here that the disruption problem is $\npconp$-complete.
Thus, at least in terms of its theoretical complexity class, it is easier than the coverage problem. In Sections~\ref{CGs with a fixed number of agents} and~\ref{CGs with a fixed number of objectives}, we complete the picture with a parameterized-complexity analysis of the problem.

We start with upper bounds. Consider a CG $\G = \tuple{G, k,\beta}$, and consider a two-player game played on $G$. For every strategy $f_2$ for \ptwo, let $\Delta_{f_2}\subseteq 2^\beta$ be the set of maximal subsets $\delta\subseteq \beta$ such that \pone has a strategy that ensures the satisfaction of the All objective $\delta$ when \ptwo uses $f_2$. That is, $\Delta_{f_2}$ is the set of maximal sets $\delta\subseteq \beta$ for which there exists a strategy $f_1$ for \pone such that $\sat(\outcome(\tuple{f_1,f_2}),\beta) = \delta$. 

\begin{lem}
\label{Delta f2}
Consider a strategy $f_2$ for \ptwo. For every strategy $f_1$ for \pone, there exists $\delta\in \Delta_{f_2}$ such that $\outcome(\tuple{f_1,f_2})$ satisfies the All-$\tilde{\gamma}$ objective $(\beta \setminus \delta)$.
\end{lem}
\begin{proof}
Consider a strategy $f_2$ for \pdis, and a strategy $f_1$ for \pcov, and let $\rho = \outcome(\tuple{f_1,f_2})$. By the definition of $\Delta_{f_2}$, there exists $\delta\in\Delta_{f_2}$ such that $\sat(\rho, \beta)\subseteq \delta$. Accordingly, $(\beta \setminus \delta) \cap \sat(\rho, \beta)=\emptyset$, and so $\rho$ satisfies all the $\tilde{\gamma}$ objectives in $\beta\setminus \delta$.
\end{proof}

\noindent
Back to CGs, the intuition behind the definition of $\Delta_{f_2}$ is that $\delta \in \Delta_{f_2}$ iff whenever \pdis follows $f_2$, \pcov can allocate one agent that covers exactly all the objectives in $\delta$. Formally, we have the following, which follows directly from the definitions.
\begin{lem}
\label{disrupting strategy Delta f2}
Consider a CG $\G = \tuple{G,k,\beta}$. A strategy $f_2$ for \pdis is a disrupting strategy in $\G$ iff for every $k$ sets $\delta_1,\ldots,\delta_k \in \Delta_{f_2}$, we have that $\bigcup_{i\in[k]} \delta_i \neq \beta$.
\end{lem}
\begin{proof}
A strategy $f_2$ is a disrupting strategy iff for every strategy $F_1$ for \pcov, there exists an objective $\alpha_i\in\beta$ such that $\alpha_i$ is not satisfied in $\rho$, for every play $\rho\in\outcome(\tuple{F_1,f_2})$. Since the sets in $\Delta_{f_2}$ are the maximal sets of objectives that \pcov can satisfy while \pdis uses $f_2$, the above claim follows.
\end{proof}

\noindent
We use the observations about $\Delta_{f_2}$ in order to restrict the search for disrupting strategies. We start with \buchi CGs.  

\begin{thm}
\label{P2 memoryless}
\pdis has a disrupting strategy in a \buchi CG $\G$ iff she has a memoryless disrupting strategy in $\G$.
\end{thm}
\begin{proof}
Consider a \buchi CG $\G = \tuple{G,k,\beta}$. 
Assume \pdis has a disrupting strategy $f_2$ in $\G$. 
By Lemma~\ref{Delta f2}, for every strategy $f_1$ for an agent of \pcov, there exists $\delta\in\Delta_{f_2}$ such that $\outcome(\tuple{f_1,f_2})$ satisfies the AllC objective $\beta\setminus \delta$, which is equivalent to the co-\buchi objective $\cup(\beta\setminus \delta)$. Accordingly, $f_2$ is a winning strategy for \ptwo in a game played on $G$ in which her objective is the ExistsC objective $\set{\cup(\beta \setminus \delta)}_{\delta\in \Delta_{f_2}}$. Let $f'_2$ be a memoryless winning strategy for \ptwo in this game. Since ExistsC objectives require memoryless strategies, such a strategy 
$f'_2$ 
exists.
We prove that $f'_2$ is a disrupting strategy 
for \pdis 
in $\G$.

Consider a strategy $F_1 = \tuple{f_1^1,\ldots,f_1^k}$ for \pcov. For every $i\in[k]$, let $\rho_i = \outcome(\tuple{f_1^i,f'_2})$. Since $f'_2$ is a winning strategy for the ExistsC objective $\set{\cup(\beta\setminus \delta)}_{\delta\in \Delta_{f_2}}$, for every $i \in [k]$ there exist $\delta_i \in \Delta_{f_2}$ such that $\rho_i$ satisfies the co-\buchi objective $\cup(\beta\setminus \delta_i)$, and so the set of \buchi objectives that $\rho_i$ satisfies is a subset of $\delta_i$. Since $\delta_1,\ldots,\delta_k \in \Delta_{f_2}$ and $f_2$ is a disrupting strategy in $\G$, then by Lemma~\ref{disrupting strategy Delta f2}, we have that $\bigcup_{i\in[k]}\delta_i \neq \beta$. Thus, there exists an objective $\alpha_j\in\beta$ such that for every $i\in[k]$, $\alpha_j$ is not satisfied in $\rho_i$. Accordingly, we have that $f'_2$ is a disrupting strategy for \pdis in $\G$.
\end{proof}

\noindent
The fact we can restrict attention to memoryless disrupting strategies leads to an upper bound for the disruption problem, which we show to be tight. 

\begin{thm}
\label{decide disrupt buchi upper}
The disruption problem for \buchi CGs is in $\npconp$.
\end{thm}
\begin{proof}
Consider a \buchi CG $\G = \tuple{G,k,\beta}$.
An NP algorithm that uses a co-NP oracle guesses a memoryless strategy $f_2$ for \pdis, and then checks that \pcov does not have a covering strategy in the one-player game defined over the sub-graph $G_{f_2}$ induced by $f_2$. That is, the sub-graph of $G$ in which the edges from vertices in $V_2$ agree with $f_2$. Since \pcov essentially owns all the vertices in $G_{f_2}$, then by Theorem~\ref{V2 = emptyset}, deciding the existence of covering strategies in $G_{f_2}$ can be done in NP, implying the $\npconp$-upper bound.
\end{proof}

\noindent
We continue to co-\buchi CGs. Here, disrupting strategies need not be memoryless, and in fact need not be polynomial. Accordingly, the proof is more complicated and is based on a symbolic description of disrupting strategies. 

A {\em superset objective} is a set $\F\subseteq 2^{2^{V}}$ of AllB objectives, and it is satisfied in a play $\rho$ iff at least one AllB objective in $\F$ is satisfied in $\rho$ \cite{HD05}. By \cite{HD05}, every superset objective has an equivalent AllB objective. Accordingly, in two-player superset games, \ptwo has an ExistsC objective, thus \ptwo wins iff she has a memoryless winning strategy. Since a memoryless strategy for \ptwo can be checked in polynomial time and the game is determined, the problem of deciding whether there exists a winning strategy for a superset objective is co-NP-complete \cite{HD05}.

We start with an easy characterization of winning strategies for AllB objectives. 
\begin{lem}
\label{AllB prop}
Consider an AllB game $\G = \tuple{G,\delta}$, with $\delta = \set{\alpha_1,\ldots,\alpha_m}$. \pone wins $\G$ iff there exist a set $W\subseteq V$ of vertices and a set $\set{g_1^i: i\in [m]}$ of memoryless strategies for \pone, such that \pone can force the play to reach $W$, and for every vertex $v\in W$ and $i\in[m]$, the strategy $g_1^i$ forces each play from $v$ to reach $\alpha_i \cap W$, while staying in $W$.
\end{lem}
\noindent
We say that $W$ is a {\em winning cage} for $\delta$, and that $\set{g_1^i: i\in [m]}$ are its {\em witnesses}.

Lemma~\ref{inf successors} below suggests a way to describe winning strategies for AllB objectives by the sets of successors a winning strategy chooses infinitely often from every vertex.

For a set of vertices $U\subseteq V$, a strategy $f_1$ for \pone is a {\em $U$-trap} if it forces the play to reach and stay in $U$. That is, if $\textit{inf}(\outcome(\tuple{f_1,f_2}))\subseteq U$, for every strategy $f_2$ for \ptwo.

A {\em fairness requirement} for \pone is a function $g: V_1 \rightarrow 2^V$ that maps each vertex $v\in V_1$ to a subset of its successors. We say that a strategy $f_1$ for \pone$ ${\em respects $g$} if for every vertex $v\in V_1$, the strategy $f_1$ only chooses successors for $v$ from $g(v)$, and when $v$ is visited infinitely often, then $f_1$ proceeds to each of the successors of $v$ in $g(v)$ infinitely often. That is, for every vertex $v\in V_1$ and a prefix of a play $h\in V^*$ we have that $f_1(h\cdot v)\in g(v)$, and for every strategy $f_2$ for \ptwo and a vertex $v\in V_1$, if $v\in \textit{inf}(\outcome(\tuple{f_1,f_2}))$, then $g(v)\subseteq \textit{inf}(\outcome(\tuple{f_1,f_2}))$.

\begin{lem}
\label{inf successors}
Consider an AllB game $\G = \tuple{G, \delta}$, with $\delta =\set{\alpha_1,\ldots,\alpha_m}$. \pone wins $\G$ iff there exists a set of vertices $U\subseteq V$ and a fairness requirement $g$ for \pone such that every $U$-trap strategy for \pone that respects $g$ is a winning strategy for \pone in $\G$, and such a strategy exists.
\end{lem}
\begin{proof}
For the first direction, assume that there exist a set of vertices $U \subseteq V$ and a fairness requirement $g$ for \pone such that there exists a $g$-respecting $U$-trap strategy for \pone, and every such strategy is a winning strategy for \pone in $\G$. Then, \pone has a winning strategy in $\G$.

For the other direction, assume that \pone wins $\G$. By Lemma~\ref{AllB prop}, there exist a minimal winning cage $W \subseteq V$ and witnesses $\set{g_1^i : i\in [m]}$ for $W$. Let $g$ be the fairness requirement for \pone with $g(v)=\set{g_1^i(v) : i \in[m]}$, for every $v\in V_1\cap W$.

We show that there exists a $g$-respecting $W$-trap strategy for \pone, and that every such strategy is a winning strategy for \pone in $\G$.

First, note that $g$-respecting $W$-trap strategies for \pone exist, since $W$ is a winning cage for $\delta$, and for every vertex $v\in V_1\cap W$, we have that $g(v)\subseteq W$.

Next, assume by contradiction that there exists a $g$-respecting $W$-trap strategy $f_1$ for \pone that is not a winning strategy. Thus, there exists a strategy $f_2$ for \ptwo such that $U=\textit{inf}(\outcome(\tuple{f_1,f_2}))$ does not intersect all the sets in $\delta$. That is, there exists $i\in[m]$ such that $U\cap\alpha_i=\emptyset$. We describe a strategy $f^*_1$ for \pone in $\G$ such that $f^*_1$ is winning, yet there is a strategy $f'_2$ for \ptwo such that, for $U'=\textit{inf}(\outcome(\tuple{f^*_1,f'_2}))$, we have $U'\subseteq U$. Thus, $U'\cap\alpha_i=\emptyset$, contradicting the fact that $f^*_1$ is a winning strategy.

We define $f^*_1$ as follows. First, $f^*_1$ forces the play to reach $W$, and, starting with $i=1$, it follows $g_1^i$ until the play reaches $\alpha_i$, where it switches to follow $g_1^{(i+1)\mod m}$. It is easy to see that $f^*_1$ is a winning strategy for \pone in $\G$. Indeed, $W$ is a winning cage and $\set{g_1^i\in[m]}$ are witnesses for $W$.

Note that when \pone uses the strategy $f^*_1$, there exists a strategy $f'_2$ for \ptwo that makes the play reach $U$, and then chooses, for every vertex $v\in U\cap V_2$, a successor in $U$. Indeed, \ptwo can make the play reach $U$, as otherwise $f^*_1$ is a $W$-trap winning strategy that forces the play to avoid $U$, which contradicts the fact that $W$ is a minimal winning cage. Also, every vertex $v\in U\cap V_2$ has a successor in $U$, since $U$ is exactly the set of vertices that appear infinitely often in $\outcome(\tuple{f_1,f_2})$.

To complete the proof, we show that once the play reaches $U$, it stays in $U$, and thus $U'\subseteq U$. Consider a vertex $v\in U\cap V_1$. Since $v$ is visited infinitely often in $\outcome(\tuple{f_1,f_2})$ and $f_1$ respects $g$, all the vertices in $g(v)$ are visited infinitely often. Thus, $g(v)\subseteq U$. By the definition of $f^*_1$, whenever the play is at $v$, the strategy $f^*_1$ chooses a successor in $g(v)$, and therefore a successor in $U$. For vertices in $U\cap V_2$, the strategy $f'_2$ chooses a successor in $U$ by definition. Hence, once the play reaches $U$, it remains in $U$, implying that $U'\subseteq U$, and we are done.
\end{proof}

By Lemma~\ref{inf successors}, we do not need to guess an explicit strategy for \pone in order to determine whether she wins an AllB or a superset game. Instead, we can guess a set of vertices that might appear infinitely often, and the set of successors we choose infinitely often from each of those vertices.

Back to CGs, we show that disrupting strategies can be described symbolically using a fairness requirement.

We say that a set of subsets of objectives $\Delta\subseteq 2^\beta$ is {\em $k$-wise intersecting} if the intersection of every $k$ sets in $\Delta$ is not empty. That is, $\bigcap_{i\in[k]} \delta_i \neq \emptyset$, for every $\delta_1,\ldots,\delta_k\in\Delta$. 

\begin{lem}
\label{disrupting strategy iff k-wise intersecting}
A strategy $f_2$ for \pdis is a disrupting strategy iff the set $\set{(\beta\setminus \delta):\delta\in\Delta_{f_2}}$ is $k$-wise intersecting.
\end{lem}
\begin{proof}
Assume first that $f_2$ is a disrupting strategy, and consider $k$ sets $\delta_1,\ldots,\delta_k \in \Delta_{f_2}$. By the definition of $\Delta_{f_2}$, there exist strategies $f_1^1,\ldots,f_1^k$ for the agents of \pcov with $\sat(\outcome(\tuple{f_1^i,f_2}),\beta) = \delta_i$, for every $i\in[k]$. 
Since $f_2$ is a disrupting strategy, then, by Lemma~\ref{disrupting strategy Delta f2}, we have that $\bigcup_{i\in[k]} \delta_i \neq \beta$. Therefore, $\bigcap_{i\in[k]}(\beta\setminus \delta_i) \neq \emptyset$, and so $\set{(\beta\setminus \delta):\delta\in\Delta_{f_2}}$ is $k$-wise intersecting.

For the second direction, consider a strategy $f_2$ for \pdis such that $\set{(\beta\setminus \delta):\delta\in\Delta_{f_2}}$ is $k$-wise intersecting. Then, for every strategy $F_1 = \tuple{f_1^1,\ldots,f_1^k}$ for \pcov and sets $\delta_1,\ldots,\delta_k\in \Delta_{f_2}$ for which $\sat(\outcome(f_1^i,f_2))\subseteq \delta_i$, for every $i\in[k]$, there exists $\alpha_j\in \bigcap_{i\in[k]}(\beta\setminus \delta_i)$. Thus, $\alpha_j$ is not covered in $\outcome(\tuple{F_1,f_2})$. Accordingly, $f_2$ is a disrupting strategy.
\end{proof}

\begin{lem}
\label{disrupting iff fair}
If \pdis has a disrupting strategy in a co-\buchi CG, then there exist a set of vertices $U$ and a fairness requirement $g$ such that every $U$-trap strategy for \pdis that respects $g$ is a disrupting strategy, and such a strategy exists.
\end{lem}
\begin{proof}
Consider a co-\buchi CG $\G = \tuple{G,k,\beta}$, and assume there exists a disrupting strategy $f_2$ for \pdis. Let $\delta_{f_2}$ be the AllB objective equivalent to the superset objective $\set{(\beta\setminus\delta): \delta\in \Delta_{f_2}}$.
Since $f_2$ is a winning strategy for the superset objective $\set{(\beta\setminus\delta): \delta\in \Delta_{f_2}}$, and thus for its equivalent AllB objective $\delta_{f_2}$, then, by Lemma~\ref{inf successors}, there exist a set of vertices $U$ and a fairness requirement $g$ such that every $U$-trap strategy $f'_2$ for \pdis that respects $g$ is a winning strategy for $\delta_{f_2}$, and hence also a winning strategy for the superset objective $\set{(\beta\setminus\delta): \delta\in \Delta_{f_2}}$. Also, such a strategy exists. By Lemma~\ref{disrupting strategy iff k-wise intersecting}, in order to show that every such strategy $f_2'$ is a disrupting strategy, it is enough to show that $\set{(\beta\setminus\delta): \delta\in \Delta_{f'_2}}$ is $k$-wise intersecting. 

Consider a $U$-trap strategy $f'_2$ that respects $g$. Note that since $f'_2$ is a winning strategy for the superset objective $\set{(\beta\setminus\delta): \delta\in \Delta_{f_2}}$, for every set $\delta\in \Delta_{f'_2}$, there exists a set $\delta'\in \Delta_{f_2}$ such that $(\beta\setminus \delta') \subseteq (\beta \setminus \delta)$.
Now, consider $k$ sets $\delta_1,\ldots,\delta_k\in \Delta_{f'_2}$ of objectives, and for every $i\in [k]$, let $\delta'_i \in \Delta_{f_2}$ be a set of objectives such that $(\beta\setminus \delta'_i)\subseteq (\beta\setminus \delta_i)$. Since $\bigcap_{i\in [k]}(\beta\setminus \delta'_i)\subseteq \bigcap_{i\in [k]}(\beta\setminus \delta_i)$ and $\set{(\beta\setminus\delta): \delta\in \Delta_{f_2}}$ is $k$-wise intersecting, we have that $\bigcap_{i\in [k]}(\beta\setminus \delta_i) \neq \emptyset$. Accordingly, $\set{(\beta\setminus\delta): \delta\in \Delta_{f'_2}}$ is $k$-wise intersecting.
\end{proof}

\noindent
The fact we can restrict attention to disrupting strategies that are given by fairness requirements leads to an upper bound for the disruption problem.

\begin{thm}
\label{decide disrupt CGC upper}
The disruption problem for co-\buchi CGs is in $\npconp$.
\end{thm}
\begin{proof}
An NP algorithm that uses a co-NP oracle guesses a set of vertices $U\subseteq V$ and a fairness requirement $g: (V_2\cap U) \rightarrow (2^U \setminus \set{\emptyset})$, checks that there exists a $U$-trap strategy for \pdis that respects $g$, and then the co-NP oracle checks that every $U$-trap strategy for \pdis that respects $g$ is a disrupting strategy. By Lemma~\ref{disrupting iff fair}, the search can be restricted to trap strategies that respect a given fairness requirement. 

We describe how the above two steps can indeed be done in NP and co-NP.
First, for the NP part, the algorithm checks that there exists a $U$-trap strategy for \pdis that respects $g$. For that, the algorithm checks that \pdis can force the play to reach $U$, checks that \pcov cannot force the play to leave $U$ by verifying that for every vertex $v\in V_1\cap U$, we have that $\succesor(v)\subseteq U$, and then checks that a $U$-trap strategy can respect $g$ by verifying that for every vertex $v\in V_2\cap U$, we have that $g(v) \subseteq \succesor(v)$. 

Then, for the co-NP oracle, the algorithm checks that every $U$-trap strategy for \pdis that respects $g$ is a disrupting strategy. For that, the algorithm checks that for every $U$-trap strategy $f_2$ for \pdis that respects $g$ and every strategy $F_1 = \tuple{f_1^1,\ldots,f_1^k}$ for \pcov, we have that $\beta$ is not covered in $\outcome(\tuple{F_1,f_2})$. That is, $\bigcap_{i\in[k]} \set{\alpha_j\in \beta: \textit{inf}(\outcome(\tuple{f_1^i,f_2}))\cap \alpha_j \neq \emptyset}\neq \emptyset$. For that, the algorithm checks that for every $k$ sets $U_1,\ldots,U_k\subseteq U$ of vertices such that for every $i\in[k]$ we have that $U_i = \textit{inf}(\outcome(\tuple{f_1,f_2}))$ for some strategy $f_1$ for an agent of \pcov and a $U$-trap strategy $f_2$ for \pdis that respects $g$, we have that 
$\bigcap_{i\in [k]} \set{\alpha_j\in\beta: U_i\cap \alpha_j \neq \emptyset}\neq \emptyset$. Note that for every set of vertices $U\subseteq V$ and a fairness requirement $g$ for \pdis, for every $U$-trap strategy $f_2$ for \pdis that respects $g$ and a strategy $f_1$ for an agent of \pcov, we have that $U' = \textit{inf}(\outcome(\tuple{f_1,f_2}))$ is a strongly connected subset of $U$ such that for every vertex $v\in V_2\cap U'$ we have that $g(v)\subseteq U'$, and the vertices in $U' \setminus \set{g(v): v\in V_2\cap U'}$ are visited infinitely often following choices made by $f_1$, and thus for every vertex $u\in U' \setminus \set{g(v): v\in V_2\cap U'}$ we have that $u\in \set{\succesor(v): v\in V_1 \cap U'}$. Thus, the co-NP oracle 
checks that for a guessed set of $k$ sets $U_1,\ldots,U_k\subseteq U$ of vertices, for every $i\in[k]$, we have that $U_i$ is strongly connected, $\set{g(v): v\in V_2\cap U_i}\subseteq U_i$, and $U_i \setminus \set{g(v): v\in V_2\cap U_i} \subseteq \set{\succesor(v): v\in V_1\cap U_i}$.
\end{proof}

\noindent
We continue to lower bounds. Note that the game in the reduction in the proof of Theorem~\ref{decide cover lower} is determined when $\Phi$ is of the form $\forall X \exists Y \phi$. Indeed, in this case the assignment to the universal variables is independent of the assignment to the existential variables, and thus if $\Phi = \bff$, which holds iff the 2QBF formula $\exists X \forall Y \overline{\phi}$ is valid, \pdis has a disrupting strategy. Since 2QBF is $\npconp$-hard, we have the following.

\begin{thm}
\label{decide disrupt lower}
The disruption problem for \buchi or co-\buchi CGs is $\npconp$-hard.
\end{thm}
\begin{proof}
We describe reductions from $2$QBF.
Consider a 2QBF formula $\Phi = \exists X \forall Y\phi$ with $\phi = C_1\vee\cdots\vee C_m$ and $C_i = (l_i^1\wedge l_i^2\wedge l_i^3)$, for every $i\in [m]$. For both $\gamma\in\set{\text{B,C}}$, we construct from $\Phi$ a $\gamma$-CG $\G$ such that $\Phi = \bft$ iff \pdis has a disrupting strategy in $\G$.

We define the $\gamma$-CG $\G$ as in the proof of Theorem~\ref{decide cover lower}, for the QBF formula $\overline{\Phi} = \forall X \exists Y \overline{\phi}$. Note that since $\phi$ is in 3DNF, the formula $\overline{\phi}$ is in 3CNF. We prove the correctness of the construction. That is, we prove that $\Phi = \bft$ iff \pdis has a disrupting strategy in $\G$.

For the first direction, assume $\Phi = \bft$. Let $\xi: X\rightarrow\set{\bft,\bff}$ be an assignment to the variables in $X$ such that for every assignment $\zeta:Y\rightarrow\set{\bft,\bff}$ to the variables in $Y$, we have that $\xi$ and $\zeta$ satisfy $\phi$. Consider a strategy $f_2$ for \pdis that chooses an assignment to the variables in $X$ according to $\xi$. That is, for every variable $x\in X$, the strategy $f_2$ proceeds to the literal vertex $x$ if $\xi(x) = \bft$, and otherwise proceeds to the literal vertex $\overline{x}$. It is not hard to see that $f_2$ is a disrupting strategy for \pdis in $\G$. Indeed, for every strategy $F_1$ for \pcov that allocates one agent to stay in the literal vertices of each variable, there exists a clause $\overline{C_{i}}$ of $\overline{\phi}$ all whose literals are evaluated to $\bff$ in the corresponding assignment, thus the objective $\alpha_{\overline{C_{i}}}$ is not covered in $\outcome(\tuple{F_1,f_2})$.

For the second direction, assume $\Phi = \bff$. Then, $\overline{\Phi} = \bft$, and so \pcov has a covering strategy in $\G$, as shown in the proof of Theorem~\ref{decide cover lower}. Consequently, \pdis does not have a disrupting strategy in $\G$.
\end{proof}

\section{Parameterized Complexity of Coverage Games}
\label{pc cg}
Consider a CG $\G = \tuple{G, k, \beta}$. 
In typical applications, the game graph $G$ is much bigger than $k$ and $\beta$. It is thus interesting to examine the complexity of CGs in terms of the size of $G$, assuming the other parameters are fixed. In this section we study the parameterized complexity of deciding \buchi and co-\buchi CGs. Note that beyond the fact that $G$ is much bigger than $k$ and $\beta$, fixing $G$ also fixes the size of $\beta$, and that fixing $\beta$ also fixes the interesting cases of the number of agents. Accordingly, in Section~\ref{CGs with a fixed number of agents} we study 
CGs with a fixed number of agents, and in Section~\ref{CGs with a fixed number of objectives} we study CGs with a fixed number of objectives. 

\subsection{CGs with a fixed number of agents}
\label{CGs with a fixed number of agents}
In this section we study the complexity of deciding \buchi and co-\buchi CGs when the number of agents is fixed.

We start with the coverage problem.
\begin{thm}
\label{2agent cover buchi np hard}
The coverage problem for \buchi or co-\buchi CGs with a fixed number of agents is NP-complete. Hardness in NP holds already for CGs with two agents.  
\end{thm}
\begin{proof}
We start with the upper bound. When the number of agents is fixed, the depth of the recursion performed by the NTM described in the proof of Theorem~\ref{decide cover buchi upper} is fixed, and so the NTM runs in polynomial time. Hence, deciding whether \pcov has a covering strategy \buchi or co-\buchi CG with a fixed number of agents can be done in NP. 

We continue to the lower bounds. 
In Theorem~\ref{V2 = emptyset and fixed k C}, we show that deciding whether \pcov has a covering strategy in a co-\buchi CG with $k=2$ is NP-hard already for the case $V=V_1$.

For \buchi CGs, we describe a reduction from $3$SAT. Consider a set of variables $X = \set{x_1,\ldots,x_n}$, and let $\phi = C_1 \wedge \cdots \wedge C_m$ with $C_i = (l_i^1\vee l_i^2\vee l_i^3)$ and $l_i^1,l_i^2,l_i^3\in X\cup \overline{X}$, for every $i\in [m]$. We construct from $\phi$ a CG $\G = \tuple{G,2,\beta}$ such that \pcov has a covering strategy in $\G$ iff $\phi$ is satisfiable.

\begin{figure}[htp]
\centering
\includegraphics[width=10.8cm]{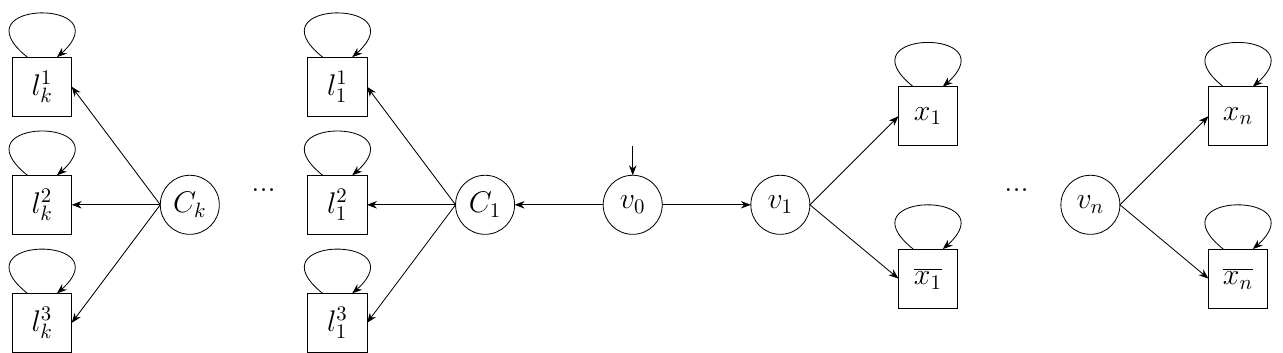}
\caption{\pcov has a covering strategy in $\G$ iff $\phi$ is satisfiable.}
\label{fig sat 2 agents}
\end{figure}

\noindent
The game graph $G$ is as follows (see Figure~\ref{fig sat 2 agents}). From the initial vertex, \pcov chooses between proceeding (right) to an {\em assignment sub-graph}, and (left) to a {\em clause sub-graph}. In the assignment sub-graph, \pcov chooses an assignment to the variables in $X$ by visiting assignment-literal vertices. The assignment-literal vertices are owned by \pdis, who chooses between staying in the assignment-literal vertex and proceeding to the next variable. In the clause sub-graph, \pcov chooses a literal for every clause of $\phi$. She does so by visiting clause-literal vertices. The clause-literal vertices are owned by \pdis, who chooses between staying in the clause-literal vertex and proceeding to the next clause. Note that plays in $G$ do not return to $v_0$, thus each play eventually gets stuck in some literal vertex. 

We define the objectives in $\G$ as follows. 
The first two objectives $\alpha_1$ and $\alpha_2$ are the sets of vertices of the assignment and clause sub-graphs, ensuring that \pcov allocates one agent to each sub-graph. We call them the {\em assignment agent} and the {\em clause agent}, respectively. In addition, for every literal $l\in X\cup \overline{X}$, we define the objective $\alpha_l$ as the union of the set of assignment-literal vertices that correspond to literals $l'\neq l$ with the set of clause-literal vertices that correspond to literals $l'\neq \overline{l}$. Note that the objective $\alpha_l$ is satisfied iff the play in the assignment sub-graph stays in an assignment-literal vertex that does not correspond to $l$ or the play in the clause sub-graph stays in a clause-literal vertex that does not correspond to $\overline{l}$. Accordingly, if the assignment agent chooses a satisfying assignment, and the clause agent chooses literals that are evaluated to $\bft$ according to chosen assignment, then every objective $\alpha_l$ is satisfied, no matter where \pdis chooses the plays to stay in the different sub-graphs. On the other hand, if the assignment agent chooses a literal $l$ and the clause agent chooses the literal $\overline{l}$, then \pdis can choose to stay in those vertices in the respective sub-graphs, in which case $\alpha_l$ is not satisfied. 

Formally, $\G = \tuple{G, 2, \beta}$ is defined as follows.
First, the game graph $G = \tuple{V_1,V_2,v_0,E}$ contains the following components.
    \begin{enumerate}
        \item $V_1 = \set{v_0} \cup \set{v_i: i\in [n]} \cup \set{C_i : i\in [m]}$. The vertices in $\set{v_i: i\in [n]}$ are {\em variable vertices}, and the vertices in $\cup \set{C_i : i\in [m]}$ are {\em clause vertices}.
        \item $V_2 = X\cup \overline{X} \cup \set{l_i^1,l_i^2,l_i^3 : i\in [m]}$. The vertices in $X\cup \overline{X}$ are {\em literal vertices}, and the vertices in $\set{l_i^1,l_i^2,l_i^3 : i\in [m]}$ are {\em clause-literal vertices}.
        \item The set of edges $E$ contains the following edges.
        \begin{enumerate}
            \item $\tuple{v_0,v_1}$.
            \item $\tuple{v_i, x_i}$ and $\tuple{v_i,\overline{x_i}}$, for every $i\in [n]$.
            \item $\tuple{l,l}$, for every $l\in X\cup \overline{X}$.
            \item $\tuple{x_i,v_{i+1}}$ and $\tuple{\overline{x_i},v_{i+1}}$, for every $i < n$.
            \item $\tuple{v_0,C_1}$.
            \item $\tuple{C_i,l_i^j}$, for every $i\in [m]$ and $j\in \set{1,2,3}$.
            \item $\tuple{l_i^j,l_i^j}$, for every $i\in [m]$ and $j\in \set{1,2,3}$.
            \item $\tuple{l_i^j,C_{i+1}}$, for every $i < m$ and $j\in \set{1,2,3}$.            
        \end{enumerate}
    \end{enumerate}
The set of objectives is $\beta = \set{\alpha_1,\alpha_2}\cup \set{\alpha_l : l\in X\cup \overline{X}}$, where $\alpha_1 = X\cup \overline{X}$, $\alpha_2 = \set{l_i^j : i\in [m], j\in \set{1,2,3}}$, and $\alpha_l = ((X\cup \overline{X})\setminus \set{l})\cup \set{l_i^j : i\in [m], j\in \set{1,2,3}, \text{ and }l_i^j\neq \overline{l}}$, for every $l\in X\cup \overline{X}$.

We prove the correctness of the construction. Note that for two plays $\rho_1$ and $\rho_2$ in the assignment and clause sub-graphs, respectively, there are literal and clause literal vertices $l$ and $l_i^j$ such that $\textit{inf}(\rho_1) = l$ and $\textit{inf}(\rho_2) = l_i^j$. Then, $\beta$ is covered in $\set{\rho_1,\rho_2}$ iff $l_i^j \neq \overline{l}$. Indeed, $\alpha_{l'}$ is satisfied in $\rho_1$ for every $l'\in (X\cup \overline{X})\setminus \set{l}$, and $\alpha_l$ is satisfied in $\rho_2$ iff $l_i^j \neq \overline{l}$.

For the first direction, assume $\phi$ is satisfiable. Let $\xi: X\rightarrow \set{\bft,\bff}$ be a satisfying assignment, and let $F_1 = \tuple{f_1^1,f_1^2}$ be the strategy for \pcov such that $f_1^1$ is a strategy for the assignment agent that chooses from every variable vertex $v_i$ the literal vertex $x_i$ iff $\xi(x_i) = \bft$, and the literal vertex $\overline{x_i}$ iff $\xi(x_i) = \bff$; The strategy $f_1^2$ for the clause agent chooses from every clause vertex $C_i$ a clause-literal vertex $l_i^j$ such that $l_i^j$ is evaluated to $\bft$ in $\xi$. It is easy to see that $F_1$ is a covering strategy. Indeed, for every strategy for \pdis, there does not exist a literal $l$ such that the play in the assignment sub-graph visits the literal vertex $l$, and the play in the clause sub-graph visits a clause-literal vertex that corresponds to $\overline{l}$.

For the second direction, assume \pcov has a covering strategy $F_1 = \tuple{f_1^1,f_1^2}$. Assume WLOG that $f_1^1$ is a strategy for the assignment agent and $f_1^2$ is a strategy for the clause agent. Let $\xi$ be the assignment to the variables in $X$ that agrees with $f_1^1$. That is, a literal $l$ is evaluated to $\bft$ in $\xi$ iff $f_1^1$ proceeds from the appropriate variable vertex to the literal vertex $l$. It is easy to see that $\xi$ satisfies $\phi$. Indeed, otherwise there exists a clause $C_i$ all whose literals are evaluated to $\bff$ in $\xi$, thus \pdis can make the play in the clause sub-graph to stay in the clause-literal vertex $l_i^j$ the clause agent chooses from $C_i$, and make the play in the assignment sub-graph to stay in the literal vertex $l$ such that $l_i^j = \overline{l}$. When \pdis uses such a strategy, we have that $\alpha_l$ is not covered, contradicting the fact $F_1$ is a covering strategy. 
\end{proof}

\noindent
We continue to the disruption problem.  

\begin{thm}
\label{2agent disrupt buchi np-c}
The disruption problem for \buchi CGs with a fixed number of agents is NP-complete. Hardness in NP holds already for two agents.
\end{thm}
\begin{proof}
We start with the upper bound. 
Consider a \buchi CG $\G=\zug{G,k,\beta}$. 
An NP algorithm guesses a memoryless strategy $f_2$ for \pdis, and then checks that \pcov does not have a covering strategy in the one-player game defined over the sub-graph $G_{f_2}$ induced by $f_2$. By Theorem~\ref{P2 memoryless}, it is sufficient to consider memoryless strategies for \pdis. By Theorem~\ref{V2 = emptyset and fixed k B}, verifying the strategy against a fixed number of agents can be done in NLOGSPACE, and hence also in polynomial time.

For the lower bound, we describe a reduction from $3$SAT. Consider a set of variables $X = \set{x_1,\ldots,x_n}$, and consider $\phi = C_1 \wedge \cdots \wedge C_m$ with $C_i = (l_i^1\vee l_i^2\vee l_i^3)$ and $l_i^1,l_i^2,l_i^3\in X\cup \overline{X}$, for every $i\in [m]$. We construct from $\phi$ a CG $\G=\zug{G',2,\beta}$ such that \pdis has a disrupting strategy in $\G$ iff $\phi$ is satisfiable.

Recall the game graph $G$ from the proof of Theorem~\ref{2agent cover buchi np hard}. The game graph $G'$ is similar to $G$, except that the roles of the players in the assignment and clause sub-graphs are dualized. That is, \pdis owns the variable and clause vertices, and \pcov owns the literal vertices. The vertex $v_0$ is still owned by \pcov.

We define the objectives in $\G$ as follows. Similar to the CG in the proof of Theorem~\ref{2agent cover buchi np hard}, we define objectives $\alpha_1$ and $\alpha_2$ that 
force \pcov to allocate one agent to each sub-graph. In addition, for every literal $l \in X \cup\overline{X}$, we define the objective $\alpha_l$ to be the union of the set of assignment-literal vertices that correspond to literals $l'\neq l$ with the set of clause-literal vertices that correspond to $\overline{l}$. Note that the objective $\alpha_l$ is satisfied iff the play in the assignment sub-graph stays in an assignment literal vertex that does not correspond to $l$, or the play in the clause sub-graph stays in a clause-literal vertex that corresponds to $\overline{l}$. Accordingly, if \pdis chooses a satisfying assignment in the assignment sub-graph, and chooses literals in the clause sub-graph that are evaluated to $\bft$ according to chosen assignment, then there exists an objective $\alpha_l$ that is not satisfied, no matter where \pcov chooses the plays to get trapped in the different sub-graphs. Also, if \pdis chooses a literal $l$ in the assignment sub-graph and the literal $\overline{l}$ in the clause sub-graph, \pcov can choose to stay in those vertices in the respective sub-graphs, in which case $\alpha_l$ is satisfied. 
Thus, $\phi$ is satisfiable iff \pdis has a disrupting strategy.

Formally, the set of objectives is $\beta = \set{\alpha_1,\alpha_2}\cup \set{\alpha_l : l\in X\cup \overline{X}}$, where $\alpha_1 = X\cup \overline{X}$, $\alpha_2 = \set{l_i^j : i\in [m], j\in \set{1,2,3}}$, and $\alpha_l = ((X\cup \overline{X})\setminus \set{l})\cup \set{l_i^j : i\in [m], j\in \set{1,2,3}, \text{ and }l_i^j= \overline{l}}$, for every $l\in X\cup \overline{X}$.

We prove the correctness of the construction. Note that for two plays $\rho_1$ and $\rho_2$ in the assignment and clause sub-graphs, respectively, there are assignment-literal and clause-literal vertices $l$ and $l_i^j$ such that $\textit{inf}(\rho_1) = l$ and $\textit{inf}(\rho_2) = l_i^j$. Then, $\beta$ is covered in $\set{\rho_1,\rho_2}$ iff $l_i^j = \overline{l}$. Indeed, $\alpha_{l'}$ is satisfied in $\rho_1$ for every $l'\in (X\cup \overline{X})\setminus \set{l}$, and $\alpha_l$ is satisfied in $\rho_2$ iff $l_i^j = \overline{l}$.

It is easy to see that $\phi$ is satisfiable iff \pdis has a disrupting strategy. Indeed, when \pdis only chooses literal vertices that correspond to literals evaluated to $\bft$ by a satisfying assignment, we have that $l_i^j \neq \overline{l}$, for every literal vertices $l$ and $l_i^j$ in the two sub-graphs in which the plays get trapped, thus $\alpha_l$ is not satisfied. For the second direction, every disrupting strategy for \pdis chooses an assignment in the assignment sub-graph and clause-literal vertices in the clause sub-graph such that every clause-literal \pdis chooses corresponds to a literal that is evaluated to $\bft$ by the chosen assignment, as then $l_i^j\neq \overline{l}$, for every literal vertex $l$ and clause-literal vertex $l_i^j$ \pcov chooses to stay indefinitely in. 
\end{proof}

\noindent
Although fixing the number of agents significantly reduces the complexity of the disruption problem for $\gamma = B$, the complexity remains surprisingly unchanged for $\gamma = C$, even when there are only two agents. Formally, we have the following.

\begin{thm}
\label{disrupt co-buchi fixed k}
The disruption problem for co-\buchi CGs with a fixed number of agents is $\npconp$-complete. Hardness in $\npconp$ holds already for two agents.
\end{thm}
\begin{proof}
The $\npconp$ upper bound follows from Theorem~\ref{decide disrupt CGC upper}.
For the lower bound, we describe a reduction from 2QBF. Consider sets of variables $X = \set{x_1,\ldots,x_n}$ and $Y = \set{y_1,\ldots,y_m}$, and let $\overline{X} = \set{\overline{x_1},\ldots,\overline{x_n}}$ and $\overline{Y} = \set{\overline{y_1},\ldots,\overline{y_m}}$. Let $Z = X\cup Y$ and $\overline{Z} = \overline{X}\cup \overline{Y}$. 
Consider a 2QBF formula $\Phi = \exists X\forall Y \phi$ with $\phi = C_1 \vee \cdots\vee C_k$ and $C_i = (l_i^1\wedge l_i^2\wedge l_i^3)$ and $l_i^1, l_i^2, l_i^3\in Z\cup \overline{Z}$, for every $i\in [k]$. We construct from $\Phi$ a co-\buchi CG $\G = \tuple{G,2,\beta}$ such that $\Phi = \bft$ iff \pdis has a disrupting strategy in $\G$.

The game graph $G$ (see Figure~\ref{fig co buchi disrupt 2 agents}) contains an {\em assignment} sub-graph, and a {\em refute} sub-graph. In the assignment sub-graph, \pcov chooses a variable $x_i\in X$, \pdis chooses an assignment to $x_i$, and then \pcov repeatedly visits the literal vertex that corresponds to the assignment and chooses an assignment to the variables in $Y$. Note that while in a given play, \pdis chooses an assignment to a single variable in $X$, a strategy for \pdis in the assignment sub-graph induces an assignment to all variables in $X$. In the refute sub-graph, \pcov tries to refute the assignment \pdis chooses for the variables in $X$. For that, \pcov chooses an assignment to the variables in $Y$, then \pdis chooses a clause $C_i$ of $\phi$, essentially claiming that all its literals are evaluated to $\bft$ in the joint assignment, and then \pcov chooses a literal $l_i^j$ of $C_i$, essentially claiming that $l_i^j$ is evaluated to $\bff$ in the joint assignment. If $l_i^j$ is an $X$-literal, the game proceeds to a self-looped sink that corresponds to $\overline{l_i^j}$. If $l_i^j$ is a $Y$-literal, the game proceeds to a vertex that corresponds to $\overline{l_i^j}$, and returns to the initial vertex in the refute sub-graph.

\begin{figure}[htp]
\centering
\includegraphics[width=9.5cm]{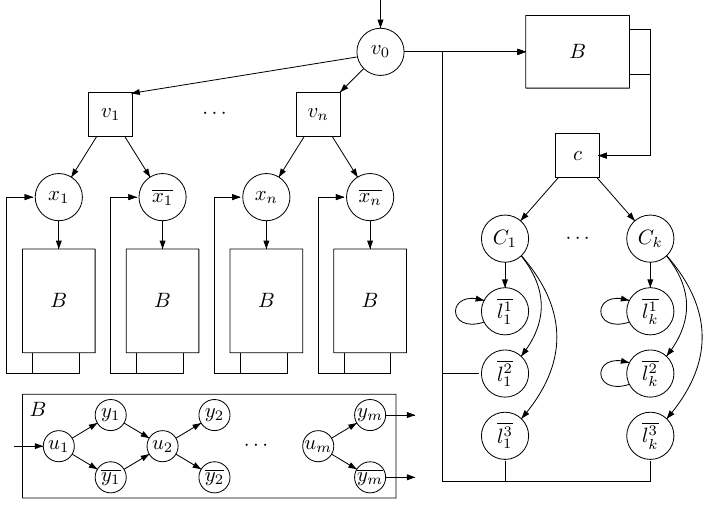}
\caption{The game graph $G$, when $l_1^1, l_k^1,l_k^2\in X\cup \overline{X}$. The sub-graph $B$ appears on the bottom (left).
}
\label{fig co buchi disrupt 2 agents}
\end{figure}

\noindent
The objectives are defined so that \pcov is forced to allocate one agent to the assignment sub-graph (the assignment agent) and one agent to the refute sub-graph (the refute agent).  
In addition, for every literal $l\in Y\cup \overline{Y}$, we define an objective $\alpha_l$ that is satisfied in plays that visit finitely often vertices that correspond to $l$ in both sub-graphs. Finally, for every literal $l\in X\cup \overline{X}$, we define an objective $\alpha_l$ that is satisfied in the refute sub-graph if the play there visits finitely often vertices that correspond to $l$, and is satisfied in the assignment sub-graph if the play there visits finitely often $X$-literal vertices that do not correspond to $l$.
Accordingly, if \pdis chooses an assignment $\xi$ to the variables in $X$ such that $\phi$ is satisfied for every assignment to the variables in $Y$, \pdis is able to always choose a clause all whose literals are evaluated to $\bft$ in the joint assignment, and thus force the refute agent to proceed to a vertex that corresponds to a literal $l$ evaluated to $\bff$ in the joint assignment. If $l$ is an $X$-literal, the objective $\alpha_l$ cannot be covered, because \pdis chooses $\overline{l}$ in the assignment sub-graph. Otherwise, there exists a $Y$-literal $l$ such that the refute agent chooses both vertices that correspond to $l$ and $\overline{l}$ infinitely often. 
Note that then, one of the objectives $\alpha_l$ and $\alpha_{\overline{l}}$ cannot be covered, as the assignment agent is forced to visit vertices that correspond to $l$ or $\overline{l}$ infinitely often. 

On the other hand, if \pdis chooses an assignment to the variables in $X$ for which there exists an assignment to the variables in $Y$ such that the joint assignment does not satisfy $\phi$, then the refute agent can choose this assignment, and thus she can choose for each clause a literal evaluated to $\bff$ in the joint assignment, in which case the refute agent only visits vertices that correspond to literals that are evaluated to $\bft$. In this case, \pcov has a strategy that, together with the strategy for \pdis, covers all the objectives.

Formally, we define $\G = \tuple{G,2,\beta}$ as follows. First, the game graph $G = \tuple{V_1,V_2,v_0,E}$ contains the following components.
    \begin{enumerate}
        \item $V_1 = \set{v_0} \cup (X\cup \overline{X}) \cup \set{u_i^l, y_i^l, \overline{y_i}^l: l\in X\cup \overline{X}, i\in[m]} \cup \set{u_i:i\in[m]}\cup (Y\cup \overline{Y})\cup \set{C_i: i\in[k]}\cup \set{\overline{l_i^j}: i\in[k], j\in\set{1,2,3}}$. The vertices in $(X\cup \overline{X})$ are {\em $X$-literal vertices}, the vertices in $\set{u_i^l:i\in[m]}$ are {\em the $Y$-variable vertices of $l$} and the vertices in $\set{y_i^l,\overline{y_i}^l:i\in[m]}$ are {\em the $Y$-literal vertices of $l$}, for every literal $l\in(X\cup \overline{X})$, the vertices in $\set{u_i,i:\in[m]}$ are {\em $Y$-variable vertices}, the vertices in $(Y\cup \overline{Y})$ are {\em $Y$-literal vertices}, the vertices in $ \set{C_i: i\in[k]}$ are {\em clause vertices}, and the vertices in $\set{\overline{l_i^j}: i\in[k], j\in\set{1,2,3}}$ are {\em refute-literal vertices}.
        \item $V_2 = \set{v_i:i\in[n]} \cup \set{c}$. The vertices in $\set{v_i:i\in[n]}$ are {\em $X$-variable vertices}, and $c$ is a {\em clause-choosing} vertex.
        \item The set of edges $E$ consists of the following edges.
        \begin{enumerate}
            \item $\tuple{v_0,v_i}$, for every $i\in[n]$, and $\tuple{v_0,u_1}$. That is, from the initial vertex, \pcov chooses an $X$-variable, or proceed to $u_1$ which is the initial vertex of the refute sub-graph.
            \item $\tuple{v_i,x_i}$ and $\tuple{v_i,\overline{x_i}}$, for every $i\in[n]$. That is, from the $X$-variable vertex $v_i$, \pdis assigns $\bft$ to $x_i$ by proceeding to the $X$-literal vertex $x_i$, and assigns $\bft$ to $x_i$ by proceeding to the $X$-literal vertex $\overline{x_i}$.
            \item $\tuple{l,u_1^l}$, for every $l\in X\cup \overline{X}$.
            \item $\tuple{u_i^l,y_i^l}$ and $\tuple{u_i^l,\overline{y_i}^l}$, for every $l\in X\cup \overline{X}$ and $i\in[m]$.
            \item $\tuple{y_i^l,u_{i+1}^l}$ and $\tuple{\overline{y_i}^l,u_{i+1}^l}$, for every $l\in X\cup \overline{X}$ and $i\leq m-1$.
            \item $\tuple{y_m^l,l}$ and $\tuple{\overline{y_m}^l,l}$, for every $l\in X\cup \overline{X}$.
            \item $\tuple{u_i,y_i}$ and $\tuple{u_i,\overline{y_i}}$, for every $i\in[m]$.
            \item $\tuple{y_i,u_{i+1}}$ and $\tuple{\overline{y_i},u_{i+1}}$, for every $i\leq m-1$.
            \item $\tuple{y_m,c}$ and $\tuple{\overline{y_m},c}$.
            \item $\tuple{c,C_i}$, for every $i\in[k]$.
            \item $\tuple{C_i,\overline{l_i^j}}$, for every $i\in[k]$ and $j\in\set{1,2,3}$.
            \item $\tuple{\overline{l_i^j},\overline{l_i^j}}$, for every $i\in[k]$ and $j\in\set{1,2,3}$ such that $l_i^j\in X\cup \overline{X}$.
            \item $\tuple{\overline{l_i^j},u_1}$, for every $i\in[k]$ and $j\in\set{1,2,3}$ such that $l_i^j\in Y\cup \overline{Y}$.
        \end{enumerate}
    \end{enumerate}
The set of objectives is $\beta = \set{\alpha_1,\alpha_2}\cup \set{\alpha_l:l\in X\cup \overline{X}\cup Y\cup \overline{Y}}$, where 
\begin{enumerate}
\item
$\alpha_1 = X\cup \overline{X}$, 
\item
$\alpha_2 = \set{\overline{l_i^j}:i\in[k],j\in\set{1,2,3}}$. 
\item
For every literal $l\in Y\cup \overline{Y}$, we have that $\alpha_l = \set{l}\cup \set{{l^{l'}}:l'\in X\cup \overline{X}} \cup \set{\overline{l_i^j}:i\in[k], j\in \set{1,2,3},\text{ and } \overline{l_i^j}=l}$. That is, $\alpha_l$ consists of all the vertices in both sub-graphs that correspond to the literal $l$.
\item
For every literal $l\in X\cup \overline{X}$, we have that $\alpha_l = ((X\cup \overline{X})\setminus\set{l}) \cup \set{\overline{l_i^j}:i\in[k], j\in \set{1,2,3},\text{ and } \overline{l_i^j}=l}$. That is, $\alpha_l$ consists of all of the vertices in the refute sub-graph that correspond to $l$, and all the $X$-literal vertices in the assignment sub-graph that do not correspond to $l$.
\end{enumerate}

\noindent
We prove the correctness of the construction. That is, we prove that $\Phi = \bft$ iff \pdis has a disrupting strategy in $\G$. Note that the objectives $\alpha_1$ and $\alpha_2$ force \pcov to allocate one agent to each sub-graph, so we only refer to strategies for \pcov that do so, and we call the respective agents the assignment agent, and the refute agent.

For the first direction, assume $\Phi = \bft$. Consider a witness assignment $\xi$ to the variables in $X$, and let $f_2$ be a strategy for \pdis that chooses $X$-literal vertices that correspond to literals evaluated to $\bft$ in $\xi$ in the assignment sub-graph, and in the refute sub-graph chooses a clause $C_i$ all whose literals are evaluated to $\bft$ in $\xi$ and the latest assignment the refute agent chose to the variables in $Y$. We show that $f_2$ is a disrupting strategy. Consider a strategy $F_1 = \tuple{f_1^1,f_1^2}$ for \pcov, assume WLOG that $f_1^1$ is a strategy for the refute agent and $f_1^2$ is a strategy for the assignment agent. Also, let $\rho_1 = \outcome(\tuple{f_1^1,f_2})$ and $\rho_2 = \outcome(\tuple{f_1^2,f_2})$. Note that since \pdis only chooses clauses all whose literals are evaluated to $\bft$ in $\xi$ and the latest assignment to the variables in $Y$, the refute agent chooses a refute-literal vertex $\overline{l_i^j}$ such that $\overline{l_i^j}$ is evaluated to $\bff$ in this assignment.

If the play $\rho_1$ reaches a refute-literal vertex $\overline{l_i^j}$ such that $\overline{l_i^j} = l$ for some literal $l\in X\cup \overline{X}$, then $\alpha_l$ is not covered in $\outcome(\tuple{F_1,f_2})$. Indeed, since $\overline{l_i^j}\in \alpha_l$, the objective $\alpha_l$ is not satisfied in the play $\rho_1$. Also, the play in the assignment sub-graph satisfies $\alpha_l$ iff it visits finitely often $X$-literal vertices that correspond to literals $l'\neq l$, and thus visits the $X$-literal vertex $l$ infinitely often. However, $\rho_2$ visits infinitely often an $X$-literal vertex that corresponds to a literal that is evaluated to $\bft$ in $\xi$, and $l$ is evaluated to $\bff$ in $\xi$.

Otherwise, the play $\rho_1$ visits finitely often refute-literal vertices that correspond to literals in $X\cup \overline{X}$. After each time the refute agent chooses an assignment to the variables in $Y$, she is forced to visit a refute-literal vertex $\overline{l_i^j}$ such that $\rho_1$ previously visits the $Y$-literal vertex that corresponds to $l_i^j$. Thus, there exists a literal $l\in Y\cup \overline{Y}$ such that both objectives $\alpha_l$ and $\alpha_{\overline{l}}$ are not satisfied in $\rho_1$. Since $\rho_2$ satisfies at most one objective between $\alpha_l$ and $\alpha_{\overline{l}}$, we have that one of them is not covered in $\outcome(\tuple{F_1,f_2})$. Therefore, $f_2$ is a disrupting strategy.

For the second direction, assume $\Phi = \bff$. Consider a strategy $f_2$ for \pdis. Let $\xi$ be the assignment that \pdis chooses to the variables in $X$ in the assignment sub-graph, and let $\zeta$ be an assignment to the variables in $Y$ such that $\phi$ is not satisfied by $\xi$ and $\zeta$. 
We define a strategy $F_1 = \tuple{f_1^1,f_1^2}$ for \pcov such that $\beta$ is covered in $\outcome(\tuple{F_1,f_2})$, implying that $f_2$ is not a disrupting strategy.

We define $F_1$ as follows. First, the strategy $f_1^1$ for the refute agent only chooses $Y$-literal vertices that correspond to literals that are evaluated to $\bft$ in $\zeta$, and for each clause $C_i$, chooses a refute-literal vertex $\overline{l_i^j}$ such that $l_i^j$ is evaluated to $\bff$ in $\xi$ and $\zeta$, and thus $\overline{l_i^j}$ is evaluated to $\bft$ in $\xi$ and $\zeta$. 

Now, the strategy $f_1^2$ for the assignment agent is defined as follows. We distinguish between two cases. 
If $\outcome(\tuple{f_1^1,f_2})$ reaches a refute-literal vertex $\overline{l_i^j}$ such that $\overline{l_i^j}=l$ for some $l\in X\cup\overline{X}$, then we define $f_1^2$ so that it  chooses the $X$-variable vertex $x_i$ such that $l\in \set{x_i,\overline{x_i}}$. Note that since $l$ is evaluated to $\bft$ in $\xi$, the strategy $f_2$ proceeds from the $X$-variable vertex $x_i$ to the $X$-literal vertex $l$. It is not hard to see that $\beta$ is covered in $\outcome(\tuple{F_1,f_2})$. Indeed, $l\in X \cup \overline{X}$, and so for every literal $l'\in Y\cup \overline{Y} \cup ((X\cup \overline{X})\setminus \set{l})$, we have that $\alpha_{l'}$ is satisfied in the play in the refute sub-graph, and since the play in the assignment sub-graph visits finitely often $X$-literal vertices that correspond to literals in $((X\cup \overline{X})\setminus \set{l}$, we have that $\alpha_l$ is satisfied in the play in the assignment sub-graph. 

Otherwise, namely if $\outcome(\tuple{f_1^1,f_2})$ avoids refute-literal vertices that correspond to literals in $X\cup \overline{X}$, then we define $f_1^2$ so that it chooses some $X$-variable vertex, and then only chooses $Y$-literal vertices of the appropriate literal that correspond to literals evaluated to $\bff$ in $\zeta$. Note that then the play in the refute sub-graph visits finitely often vertices that correspond to literals in $Y\cup \overline{Y}$ that are evaluated to $\bff$ in $\zeta$, and the play in the assignment sub-graph visits finitely often vertices that correspond to literals in $Y\cup \overline{Y}$ that are evaluated to $\bft$ in $\zeta$. Also, for every literal $l\in X\cup \overline{X}$ we have that $\alpha_l$ is satisfied in the play in the refute sub-graph, as it visits finitely often refute-literal vertices that correspond to $l$. Thus, $\beta$ is covered in $\outcome(\tuple{F_1,f_2})$, and we are done.  
\end{proof}

\subsection{CGs with a fixed number of objectives}
\label{CGs with a fixed number of objectives}
In this section we continue the parameterized complexity analysis of  \buchi and co-\buchi CGs and analyze the complexity of CGs when the number of underlying objectives is fixed. We show that both problems are PTIME-complete, regardless of whether the number of agents is fixed too. 

Consider a CG $\G = \tuple{G, k, \beta}$. As detailed in Section~\ref{k easy}, when $k \geq |\beta|$, both the coverage and disruption in $\G$ correspond to $|\beta|$ two-player games with a single objective (see Lemma~\ref{P1 k geq m wins iff she m wins}). Since two-player \buchi and co-\buchi games are PTIME-complete \cite{VW86a,Imm81}, we have the following:

\begin{thm}
\label{fixed beta unfixed k coverage}
The coverage and disruption problems for \buchi or co-\buchi CGs with a fixed $|\beta|$ and $k \geq |\beta|$ are PTIME-complete.
\end{thm}

 \noindent
 Accordingly, for the rest of the section we assume that $k < |\beta|$, and so fixing $|\beta|$ also fixes $k$. We start with deciding whether \pcov has a covering strategy.

\begin{thm}
\label{fixed beta coverage}
The coverage problems for \buchi or co-\buchi CGs with fixed numbers of agents and underlying objectives are PTIME-complete.
\end{thm}
\begin{proof}
Consider a CG $\G = \tuple{G,k,\beta}$. Recall the upper bound proof of Theorem~\ref{decide cover buchi upper}, which describes an NTM that runs in polynomial space. The nondeterminism is used to guess forks, where each guess involves guessing a partition of $\beta$ among the $k$ agents and then recursively checking that \pcov has covering strategies in the resulting CGs. When both $\beta$ and $k$ are fixed, the number of possible partitions is constant, and the nondeterministic guesses can be replaced by a deterministic enumeration of all options. This yields a DTM that also runs in polynomial space.

Furthermore, as shown in Theorem~\ref{2agent cover buchi np hard}, fixing $k$ bounds the recursion depth by a constant. Thus, when both $k$ and $|\beta|$ are fixed, the entire procedure runs in polynomial time, and the overall complexity is PTIME.

Hardness in PTIME follows from the PTIME-hardness of \buchi and co-\buchi games \cite{Imm81}, which correspond to the special case with $k=|\beta|=1$. 
\end{proof}

\noindent
We continue to the disruption problem.

\begin{thm}
\label{fixed beta disruption}
The disruption problems for \buchi or co-\buchi CGs with fixed numbers of agents and  underlying objectives are PTIME-complete.
\end{thm}
\begin{proof}
Consider a CG $\G = \tuple{G,k,\beta}$. Recall that for a strategy $f_2$ for \pdis in $G$, the set $\Delta_{f_2} \subseteq 2^{\beta}$ consists of the maximal subsets $\beta' \subseteq \beta$ such that there exists a strategy $f_1$ for \pcov with which $\sat(\outcome(\tuple{f_1,f_2}), \beta) = \beta'$. Also recall that, by Lemma~\ref{disrupting strategy Delta f2}, a strategy $f_2$ for \pdis is a disrupting strategy in $\G$ iff for every $k$ sets $\delta_1,\ldots,\delta_k \in \Delta_{f_2}$, we have $\bigcup_{i\in[k]} \delta_i \neq \beta$.

Based on this characterization, we can construct polynomial-time algorithms for the disruption problems as follows. The algorithms iterate over sets $\Delta \subseteq 2^\beta$ such that for all $\delta_1, \ldots, \delta_k \in \Delta$, the union $\bigcup_{i \in [k]} \delta_i$ is a strict subset of $\beta$, and check whether there exists a strategy $f_2$ for \pdis such that $\Delta_{f_2} \subseteq \Delta$.

When the underlying objectives are \buchi, each such set $\Delta$ induces the ExistsC objective $\set{\cup(\beta \setminus \delta) : \delta \in \Delta}$, and the algorithm checks  whether \ptwo has a winning strategy for this objective. In the case of co-\buchi objectives, the set $\Delta$ induces the superset objective $\set{(\beta \setminus \delta) : \delta \in \Delta}$, which the algorithm translates to an equivalent AllB objective and checks whether \ptwo has a winning strategy for it.

Since $\beta$ is fixed, the number of candidate sets $\Delta$ and the size of the induced objectives are fixed, and so all steps can be performed in polynomial time.
\end{proof}

\section{Discussion}
\label{disc}
We introduced and studied coverage games, extending multi-agent planning to adversarial factors that go beyond those studied in existing frameworks. 

Below we discuss variants of CGs. Clearly, all classical extensions to the underlying two-player game (e.g., concurrency, probability, partial visibility, etc.) and to the objectives (e.g., richer winning conditions, weighted objectives, mixing objectives of different types, etc.), as well as classical extensions from planning (e.g., dynamic objectives, optimality of agents and their resources, etc.) are interesting also in the setting of CGs. For example, handling {\em reachability\/} objectives, it is interesting to search for strategies of \pcov that reach the objectives via short paths, and one may optimize the longest or the average path. 

Although the extension to richer objectives is an obvious direction for future work, we note that the study of \buchi and co-\buchi better highlights the differences between CGs and standard games, as the complexity is not dominated by challenges that have to do with the objective. For example, by Footnote~2, the extension of the general M\"{u}ller objective to CGs involves no cost, whereas for the \buchi and co-\buchi objectives, we exhibit a variety of complexity classes, obtained for the different variants of CGs.

As our results show, the transition from multiple objectives and a single agent to CGs adds the challenge of decomposing the set $\beta$ of objectives among the agents. For example, while multiple co-\buchi objectives can be merged, making them easier than multiple \buchi objectives, 
our results show that 
co-\buchi CGs are not easier 
(in fact, for the case of a fixed number of agents, even harder) 
than \buchi CGs. It is interesting to study how the need to decompose the objectives affects the problem in richer settings.  

More interesting in our context are extensions that have to do with the operation of several agents: communication among the agents and the ability of \pdis to also use different agents. Recall that in the CGs studied here, the strategy of each agent may depend on the history of her interaction with \pdis, but is independent of the interaction of the other agents.
Settings in which the agents communicate with each other correspond to CGs in which their strategies depend on the history of all interactions. 

Clearly, communication can help \pcov. For example, if two drones are tasked with patrolling an area and Drone 1 is drifting northward at a certain location, knowing this could prompt Drone 2 to avoid that location or to fly southward. 

Also, beyond full information about the other agents, various applications induce interesting special cases, such as visibility of a subset of the agents, of agents in some radius, of their current location only, and so on. 
 
 As for a richer disruption, recall that in the CGs studied here, all the agents of \pcov interact with the same strategy of \pdis. It is interesting to consider settings in which \pdis also operates a number $l$ of agents. Then, a strategy for \pdis is a vector $F_2=\zug{f_2^1,\ldots,f_2^l}$ of strategies, and \pcov should cover all the objectives in $\beta$ for every $F_2$ and for every possible pairing of each of her agents with those of \pdis. 

This extension corresponds to settings in which different agents may face different responses, even after the exact same interaction. Clearly, this can help \pdis. For example, 
if a system aims to ensure that some road remains uncongested, it may limit the number of cars in each zone by directing different cars in different directions. 
From a technical point of view, the setting involves additional challenges. In particular, decomposition now has an additional parameter. For example, it can be shown that if $l=k$, then \pcov can cover $\beta$ iff there is an a-priori decomposition of $\beta$ among the agents.

\bibliographystyle{alphaurl}
\bibliography{ok_lmcs}
\end{document}